\documentclass[10pt,a4paper]{article}
\usepackage[utf8]{inputenc}
\usepackage[T1]{fontenc}
\usepackage{lmodern}
\usepackage[a4paper,margin=2.5cm]{geometry}
\usepackage{amsmath,amssymb,amsfonts,amsthm,bm}
\usepackage{graphicx}
\usepackage{textcomp}
\usepackage{xcolor}
\usepackage{float}
\usepackage{caption}
\usepackage{booktabs}
\usepackage[colorlinks=true,linkcolor=blue,citecolor=blue,urlcolor=blue]{hyperref}
\hypersetup{bookmarksdepth=3}  % omit \paragraph from PDF bookmarks
\usepackage{xurl}
\usepackage{authblk}
\usepackage{tikz}
\usetikzlibrary{arrows.meta}

\usepackage[backend=biber, style=ieee, isbn=false, sortcites,
  maxbibnames=6, minbibnames=1, safeinputenc]{biblatex}
\DeclareSourcemap{
  \maps[datatype=bibtex]{\map{\step[fieldsource=abstract,null]\step[fieldset=abstract,null]}}
}
\newtheorem{theorem}{Theorem}
\newtheorem{corollary}[theorem]{Corollary}
\newtheorem{proposition}{Proposition}

\newtheorem{definition}{Definition}

\theoremstyle{remark}
\newtheorem{remark}{Remark}

\title{\textbf{Variational and Majorization Principles in Lattice Reduction}}
\author[1]{Javier Blanco-Romero\thanks{Contact author: \texttt{frblanco@pa.uc3m.es}}}
\author[1]{Florina Almenares Mendoza}
\affil[1]{Department of Telematic Engineering, Universidad Carlos III de Madrid, Leganés, Madrid, 28911, Spain}

\date{}

\begin{document}
\maketitle

% ===================================================================
%  ABSTRACT
% ===================================================================
\begin{abstract}
Lattice reduction smooths the Gram-Schmidt profile, and we use majorization to describe the local swap mechanism behind that smoothing. In this language, each non-degenerate Lov\'asz swap acts as a T-transform on the log-norm profile. As a consequence, every strictly Schur-convex measure of profile spread decreases at such a swap. Two structural consequences follow. First, the worst-case GSA envelope admits a variational interpretation. It is the unique minimum-variance profile compatible with the Lov\'asz gap geometry, so its slope is determined by the LLL parameter alone. Second, the realized swap trajectory satisfies an exact telescoping identity for variance dissipation. The same viewpoint also helps organize deep-insertion heuristics. It suggests a thermal family of Schur-convex scoring rules, motivates adaptive selection within that family, and leads to two concrete selectors: Thermal-Adaptive, which reduces operation counts relative to SS-GG on flat profiles in our benchmarks while recovering SS-GG on $q$-ary inputs, and Geodesic Deep-LLL, which reduces equivalent-swap counts on structured lattices in our benchmarks at higher wall-clock cost.
\end{abstract}



% ===================================================================
\section{Introduction}
\label{sec:intro}
% ===================================================================

Lattice reduction smooths a basis. The Gram-Schmidt log-norm profile of a raw lattice basis is usually jagged, and after LLL~\cite{lenstra1982factoring}, BKZ~\cite{schnorr1994lattice}, or BKZ~2.0~\cite{chen2011bkz} it is often close to linear. Two questions sit behind that familiar picture: why does the profile flatten, and how much can the next move steer that flattening? The standard answer to the first is usually phrased through the Geometric Series Assumption~\cite{schnorr2003lattice} and explored quantitatively with simulators~\cite{chen2011bkz}. Together they describe the profile shape reduction tends to produce, but not the local swap dynamics that drive the flattening. The second question has led to a family of deep-insertion selectors~\cite{fontein2014potlll,yasuda2019new,bhattacherjee2025greedy}. Each chooses an objective and descends in it, with the choice justified mainly by algorithmic behavior.

We go one level lower, to the individual swap. A non-degenerate Lov\'asz swap preserves the sum of two adjacent log-norms and pulls them strictly closer together. In the language of majorization theory~\cite{marshall1979inequalities}, it is a T-transform, so the post-swap profile is majorized by the pre-swap one. Every strictly Schur-convex function of the profile therefore decreases at every non-degenerate swap, independently of the dimension, the input distribution, and the LLL threshold~$\delta$. That local law has two consequences. On the dynamics side, the log-determinant pins the profile to an affine hyperplane and the Lov\'asz condition bounds adjacent log-norm gaps; under both constraints, the arithmetic progression whose slope depends only on $\delta$ is the unique minimum-variance shape. Its slope is the steepest descent any $\delta$-LLL output can sustain in the worst case; the Gram-Schmidt log slopes that Gama and Nguyen~\cite{gama2008predicting} measure on real bases are shallower, in line with the variance-dissipation picture below. On the algorithmic side, the same structure organizes deep-insertion selectors into two families, a symmetric one that treats profile entries equally and contains SS-GG, and a position-weighted one that contains Pot-DeepLLL. Within the thermal subfamily, adaptive selection becomes natural as the profile drifts. Section~\ref{sec:related} surveys prior work, Section~\ref{sec:prelim} sets the background, Section~\ref{sec:htheorem} proves the per-swap structure, Section~\ref{sec:algorithms} develops the selector implications, Section~\ref{sec:experiments} tests them, and Sections~\ref{sec:discussion}--\ref{sec:conclusions} close the loop.

% ===================================================================
\section{Related Work}
\label{sec:related}
% ===================================================================

Two strands of prior work bear on the per-swap picture: the evolution of the GSO profile during reduction, and the algorithmic exploitation of that evolution through objective-driven selectors.

On the dynamics side, Schnorr's Geometric Series Assumption~\cite{schnorr2003lattice} approximates BKZ output by a geometric sequence $\lVert\mathbf{b}_i^*\rVert^2/\lVert\mathbf{b}_1^*\rVert^2 = q^{i-1}$, with quotient $q$ depending on blocksize. Gama and Nguyen~\cite{gama2008predicting} measured Gram-Schmidt log slopes across LLL, DeepLLL, and BKZ and linked the slope to the root-Hermite factor. Chen and Nguyen~\cite{chen2011bkz} built the GSA into the BKZ~2.0 simulator; Hanrot, Pujol, and Stehl\'e~\cite{hanrot2011analyzing} analyzed BKZ as a dynamical system on profile space and showed convergence to a quasi-GSA fixed point; Micciancio and Walter~\cite{micciancio2016practical} pushed the predictive side further with closed-form quality estimates. At finer scale, Yu and Ducas~\cite{yu2017second} observed negative correlation between adjacent squared norms $r_i$ and $r_{i+1}$ and a self-stabilizing reduction of global profile variance; Bai, Stehl\'e, and Wen~\cite{bai2018measuring} documented the head concavity that makes BKZ output deviate from a straight line at small blocksize, a deviation that shrinks as the blocksize grows. These results characterize the shape reduction converges to and the rate of convergence at the tour level, but do not certify monotonicity of nonlinear profile functionals at the level of a single swap.

On the algorithmic side, nonsequential local reduction begins with DeepLLL~\cite{schnorr1994lattice}, where a basis vector is inserted at an earlier position whenever a generalized Lov\'asz condition is violated. The polynomial-time variants Pot-DeepLLL and SS-DeepLLL replaced that rule by descent in an explicit objective: the basis potential $\mathrm{Pot}(B) = \prod_i \lVert\mathbf{b}_i^*\rVert^{2(d-i+1)}$~\cite{fontein2014potlll} and the squared-sum $\mathrm{SS}(B) = \sum_i \lVert\mathbf{b}_i^*\rVert^2$~\cite{yasuda2019new}. Efficient GSO update formulas~\cite{yamaguchi2017explicit} made these schemes practical, and the SS objective already appeared as a preprocessing criterion in SVP sampling~\cite{fukase2015accelerated}. Yasuda et al.~\cite{yasuda2017analysis} gave conditional decrease guarantees for SS under individual LLL swaps. The X-GG framework of Bhattacherjee, Hernandez-Castro, and Moyler~\cite{bhattacherjee2025greedy} globalizes the descent idea, searching the whole basis for the insertion that most decreases an objective $X$ and obtaining the SS-GG and Pot-GG instances. Across this line, each objective is motivated by its algorithmic behavior, and there is no common principle that selects or parameterizes them from the geometry of the swap itself.

Majorization is standard in matrix analysis, quantum information, and resource theory~\cite{marshall1979inequalities}. We are not aware of prior work using it to organize the per-swap dynamics of lattice reduction or to derive deep-insertion objective families from the structure of the underlying move.

% ===================================================================
\section{Background}
\label{sec:prelim}
% ===================================================================

Three ingredients carry the rest of the paper: the GSO log-norm profile, the Lov\'asz swap, and the majorization language that describes what a swap does.

\subsection{The GSO log-norm profile}

\begin{definition}[GSO log-norm profile]
Given a lattice basis $\mathbf{B} = (\mathbf{b}_1, \ldots, \mathbf{b}_d)$, its GSO log-norm profile is $\mathbf{p} = (p_1, \ldots, p_d)$ with $p_i = \ln\lVert\mathbf{b}_i^*\rVert$ and $r_i = \lVert\mathbf{b}_i^*\rVert^2 = e^{2p_i}$, where $\mathbf{b}_i^*$ are the Gram-Schmidt vectors and $\ln$ denotes the natural logarithm (used throughout).
\end{definition}

Every basis change in lattice reduction multiplies $\mathbf{B}$ on the left by a unimodular matrix, so $|\det \mathbf{B}| = \prod_i r_i^{1/2}$ is preserved and so is its logarithm:
\begin{equation}\label{eq:det}
  \textstyle\sum_{i=1}^d p_i = \ln\lvert\det(\mathbf{B})\rvert.
\end{equation}
The \emph{LLL potential} $\Phi = \sum_{i=1}^d (d-i+1)\,p_i$ is a position-weighted sum, with earlier positions contributing more heavily. It serves as a termination certificate: $\Phi$ decreases strictly at each Lov\'asz swap and is bounded below by a lattice-dependent constant, so the algorithm performs only finitely many swaps. Unlike the log-determinant, $\Phi$ is strictly decreasing rather than conserved.

\subsection{Lov\'asz swaps}

A swap at position~$k$ is triggered when the \emph{Lov\'asz condition} is violated:
\begin{equation}\label{eq:lovasz}
  r_k < (\delta - \mu_{k,k-1}^2)\,r_{k-1},
\end{equation}
where $\mu_{k,k-1}$ is the $(k,k-1)$ Gram-Schmidt coefficient and $\delta \in (1/4,\,1]$ is the LLL threshold. The prefactor is always strictly positive, since size reduction enforces $|\mu_{k,k-1}| \leq 1/2$ and hence $\mu_{k,k-1}^2 \leq 1/4 < \delta$. The swap updates the GS norms by
\begin{equation}\label{eq:swap}
  r'_{k-1} = r_k + \mu_{k,k-1}^2\,r_{k-1}, \qquad
  r'_k     = \frac{r_{k-1}\,r_k}{r'_{k-1}},
\end{equation}
preserving $r'_{k-1}\,r'_k = r_{k-1}\,r_k$ and hence the log-sum $p'_{k-1} + p'_k = p_{k-1} + p_k$. Only $p_{k-1}$ and $p_k$ change.

The size-reduction bound also gives a $\mu$-free way to read the Lov\'asz condition. The hardest case is $|\mu_{i,i-1}| = 1/2$, which makes the right-hand side of~\eqref{eq:lovasz} smallest. So every size-reduced LLL-reduced basis satisfies the classical worst-case adjacent bound~\cite{lenstra1982factoring,nguyen2009lll}
\begin{equation}\label{eq:lll_worst_gap}
  r_i \ge \Bigl(\delta - \frac14\Bigr) r_{i-1},
  \qquad i=2,\ldots,d.
\end{equation}
In log coordinates, this says that every terminal adjacent gap is at most
\begin{equation}\label{eq:cdelta}
  c_\delta = \tfrac{1}{2}\ln\!\bigl(\delta - \tfrac{1}{4}\bigr)^{-1} > 0.
\end{equation}
Conversely, during reduction, any gap $p_{i-1}-p_i>c_\delta$ violates Lov\'asz for every size-reduced value of $\mu_{i,i-1}$. This is the worst-case boundary used later in Proposition~\ref{prop:gsa_minvar}.

\subsection{Deep insertions}

A deep insertion moves a basis vector from position $k$ to an earlier position $j < k$, shifting the block $(j,\ldots,k{-}1)$ one step to the right. It is therefore a long-range analogue of an adjacent swap: one move can resolve $k-j$ adjacent inversions at once~\cite{schnorr1994lattice}.

The admissibility test uses projections of $\mathbf{b}_k$ onto progressively smaller orthogonal complements. Let $\pi_\ell$ project onto the complement of $\mathrm{span}(\mathbf{b}_1^*,\ldots,\mathbf{b}_{\ell-1}^*)$, and write
\[
  P_\ell = \lVert\pi_\ell(\mathbf{b}_k)\rVert^2.
\]
At $\ell = k$ the projection is $\mathbf{b}_k^*$ itself, so $P_k = r_k$. Running the recursion backward from $\ell = k$ through $\ell = j$,
\[
  P_\ell = P_{\ell+1} + \mu_{k,\ell}^2\,r_\ell,
  \qquad \ell = k{-}1,\ldots,j,
\]
adds back the contribution of each GS direction. So $P_j$ is the squared norm of $\mathbf{b}_k$ after removing its components along the first $j-1$ GS directions, while $r_j = \lVert\mathbf{b}_j^*\rVert^2$ is the squared norm of the current $j$-th GS vector; the two are different in general.

\begin{definition}[Generalized Lov\'asz condition]
\label{def:gen_lovasz}
A deep insertion of $\mathbf{b}_k$ at position $j < k$ is \emph{admissible} when
\[
  P_j < \delta\,r_j,
\]
i.e., the squared projection norm of $\mathbf{b}_k$ at level $j$ is smaller by factor $\delta$ than the current squared GS norm $r_j$. In practice, one scans candidate positions $j = k{-}1,k{-}2,\ldots,1$ and tests this inequality with the already computed $P_j$. For $j = k-1$ this is equivalent to the ordinary Lov\'asz condition~\eqref{eq:lovasz}, since $P_{k-1} = r_k + \mu_{k,k-1}^2 r_{k-1}$.
\end{definition}

The same recursion gives the post-insertion GSO norms used throughout the algorithmic sections~\cite{yamaguchi2017explicit}. We will measure the cost of one deep insertion both by its insertion count and by its equivalent-swap depth $k-j$.

\subsection{Majorization and T-transforms}

\begin{definition}[T-transform~{\cite{marshall1979inequalities}}]
A \emph{T-transform} at coordinates $(i,j)$ is the operator $T_{ij,\varepsilon}$ acting on $\mathbf{x} \in \mathbb{R}^d$ with $x_i > x_j$ by
\[
  T_{ij,\varepsilon}\mathbf{x}
  = (x_1,\ldots,x_{i-1},\,x_i - \varepsilon,\,x_{i+1},\ldots,x_{j-1},\,x_j + \varepsilon,\,x_{j+1},\ldots,x_d),
  \qquad \varepsilon \in (0,\,x_i - x_j).
\]
It closes the gap between the two selected entries while preserving their sum; all other coordinates are unchanged.
\end{definition}

\begin{definition}[Majorization~{\cite{marshall1979inequalities}}]
$\mathbf{y}$ is \emph{majorized} by $\mathbf{x}$, written $\mathbf{y} \prec \mathbf{x}$, when $\mathbf{y}$ is reachable from $\mathbf{x}$ by finitely many T-transforms.
\end{definition}

\begin{definition}[Schur-convexity]
A function $F\colon\mathbb{R}^d \to \mathbb{R}$ is \emph{Schur-convex} if $\mathbf{y} \prec \mathbf{x}$ implies $F(\mathbf{y}) \leq F(\mathbf{x})$, and \emph{strictly Schur-convex} if the inequality is strict whenever $\mathbf{y} \neq \mathbf{x}$. It is \emph{(strictly) Schur-concave} if $-F$ is (strictly) Schur-convex.
\end{definition}

Every strictly Schur-convex function strictly decreases along any chain of T-transforms, and every strictly Schur-concave function on the positive simplex strictly increases. These two languages, lattice swaps and T-transforms, are about to coincide.

% ===================================================================
\section{Per-Swap Majorization Structure}
\label{sec:htheorem}
% ===================================================================

A Lov\'asz swap preserves the sum of two adjacent log-norms and, when $\mu_{k,k-1} \neq 0$, pushes them strictly closer together (Figure~\ref{fig:ttransform}). The labeled coordinates may cross, but the absolute gap shrinks. That is precisely the description of a T-transform. The variance decrease that follows is then immediate from Schur-convexity.

\begin{theorem}[Per-swap majorization]
\label{thm:variance}
Let $\mathbf{p} = (p_1, \ldots, p_d)$ be any GSO log-norm profile (no sign restriction on the~$p_i$). Every Lov\'asz swap at position~$k$ with $\mu_{k,k-1} \neq 0$ is a T-transform:
\[
  (p_{k-1},\, p_k)
  \;\longmapsto\;
  (p_{k-1} - \varepsilon,\; p_k + \varepsilon),
  \qquad
  \varepsilon \in (0,\, p_{k-1} - p_k),
\]
with no other profile entry affected. In particular $\mathbf{p}' \prec \mathbf{p}$: the post-swap profile is majorized by the pre-swap one.
\end{theorem}

\begin{proof}
The intuition is that a swap pushes the larger of the two log-norms down and the smaller one up by the same amount $\varepsilon$, keeping their sum fixed. The Lov\'asz condition guarantees $\varepsilon > 0$, and $\mu \neq 0$ guarantees $\varepsilon$ stops short of the gap.

Formally, Eq.~\eqref{eq:swap} shows that the swap changes only $(p_{k-1}, p_k)$ and preserves their sum $s = p_{k-1} + p_k$. Write $a = p_{k-1}$ and $b = p_k$. The Lov\'asz condition~\eqref{eq:lovasz} gives $r_k < r_{k-1}$, hence $a > b$. Define $\varepsilon = a - p'_{k-1}$; then automatically $p'_k = b + \varepsilon$. Again from Eq.~\eqref{eq:swap}$,$
\[
  r'_{k-1} = r_k + \mu^2 r_{k-1} < \delta r_{k-1} \le r_{k-1},
\]
so $p'_{k-1} < a$ and therefore $\varepsilon > 0$. If $\mu \neq 0$, then also $r'_{k-1} > r_k$, so $p'_{k-1} > b$ and hence $\varepsilon < a-b$. No sign condition on the $p_i$ is needed.

The swap preserves $p_{k-1} + p_k$, modifies no other entry, and satisfies $\varepsilon \in (0,\, a - b)$. It is by definition a T-transform at coordinates $(k{-}1,\, k)$, so $\mathbf{p}' \prec \mathbf{p}$.
\end{proof}

\begin{figure}[t]
\centering
\definecolor{tgreen}{HTML}{9fcf69}
\definecolor{tblue}{HTML}{33acdc}
\begin{tikzpicture}[>=Stealth, font=\small, scale=0.9]
  %--- Before swap (blue) ---
  \begin{scope}
    \draw[->] (0,0) -- (4.9,0) node[right] {$i$};
    \draw[->] (0,0) -- (0,3.3) node[above] {$p_i$};
    \fill[gray!20] (0.3,0) rectangle (0.8,1.4); \draw[gray!50] (0.3,0) rectangle (0.8,1.4);
    \fill[gray!20] (3.15,0) rectangle (3.65,2.0); \draw[gray!50] (3.15,0) rectangle (3.65,2.0);
    \fill[gray!20] (4.1,0) rectangle (4.6,0.9); \draw[gray!50] (4.1,0) rectangle (4.6,0.9);
    \fill[tblue!40] (1.25,0) rectangle (1.75,2.8);
    \draw[tblue!80!black, line width=1pt] (1.25,0) rectangle (1.75,2.8);
    \fill[tblue!40] (2.2,0) rectangle (2.7,0.6);
    \draw[tblue!80!black, line width=1pt] (2.2,0) rectangle (2.7,0.6);
    \node[below] at (1.5, 0) {$k{-}1$};
    \node[below] at (2.45, 0) {$k$};
    \node[right] at (1.75, 2.8) {$p_{k-1}$};
    \node[above] at (2.45, 0.6) {$p_k$};
    \draw[dashed, gray!70] (1.1,1.7) -- (2.75,1.7) node[right, gray!60] {$\bar{p}$};
    \node[above] at (2.45, 3.15) {\textit{before}};
  \end{scope}
  \draw[->, very thick] (5.1,1.6) -- (5.9,1.6) node[above, midway] {swap};
  %--- After swap (green) ---
  \begin{scope}[xshift=6.3cm]
    \draw[->] (0,0) -- (4.9,0) node[right] {$i$};
    \draw[->] (0,0) -- (0,3.3) node[above] {$p_i$};
    \fill[gray!20] (0.3,0) rectangle (0.8,1.4); \draw[gray!50] (0.3,0) rectangle (0.8,1.4);
    \fill[gray!20] (3.15,0) rectangle (3.65,2.0); \draw[gray!50] (3.15,0) rectangle (3.65,2.0);
    \fill[gray!20] (4.1,0) rectangle (4.6,0.9); \draw[gray!50] (4.1,0) rectangle (4.6,0.9);
    \fill[tgreen!40] (1.25,0) rectangle (1.75,2.3);
    \draw[tgreen!80!black, line width=1pt] (1.25,0) rectangle (1.75,2.3);
    \fill[tgreen!40] (2.2,0) rectangle (2.7,1.1);
    \draw[tgreen!80!black, line width=1pt] (2.2,0) rectangle (2.7,1.1);
    \node[below] at (1.5, 0) {$k{-}1$};
    \node[below] at (2.45, 0) {$k$};
    \node[right] at (1.75, 2.3) {$p'_{k-1}$};
    \node[above] at (2.45, 1.1) {$p'_k$};
    \draw[dashed, gray!70] (1.1,1.7) -- (2.75,1.7) node[right, gray!60] {$\bar{p}$};
    % k-1: bar decreased, arrow pointing DOWN (left of bar)
    \draw[->, red!70!black, thick] (1.1,2.8) -- (1.1,2.3);
    \node[red!70!black, anchor=east, font=\small] at (1.08,2.55) {$\varepsilon$};
    % k: bar increased, arrow pointing UP (right of bar)
    \draw[->, red!70!black, thick] (2.85,0.55) -- (2.85,1.05);
    \node[red!70!black, anchor=west, font=\small] at (2.87,0.80) {$\varepsilon$};
    \node[above] at (2.45, 3.15) {\textit{after}};
  \end{scope}
\end{tikzpicture}
\caption{A Lov\'asz swap as a T-transform on the GSO log-norm profile (schematic, $d{=}5$, $k{=}3$). The two highlighted entries move by~$\varepsilon$ toward their mean~$\bar{p}$: position~$k{-}1$ moves down, position~$k$ moves up, and all other entries stay fixed. The local sum is preserved and the gap closes strictly, so $\mathbf{p}' \prec \mathbf{p}$.}
\label{fig:ttransform}
\end{figure}
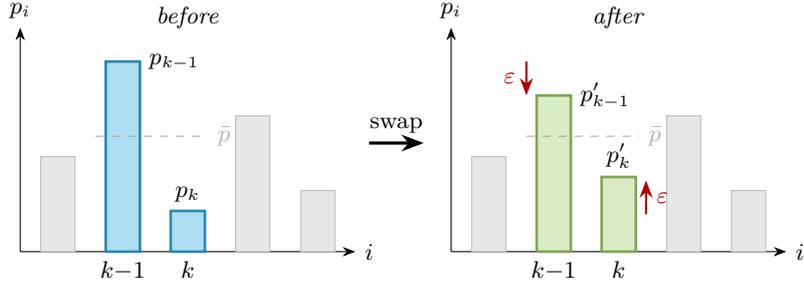

Once a non-degenerate swap is identified as a T-transform, monotonicity of Schur-convex profile functionals follows from the standard majorization theorem, with strict decrease for strictly Schur-convex functionals. When the normalized profile lies in the positive simplex, Schur-concave entropies move in the opposite direction for the same reason. We spell out one example of each. Variance is the natural $\ell^2$ measure of spread, and Shannon entropy ties the picture to information theory.

\begin{corollary}[Sum-of-squares decrease]
\label{cor:variance}
At every Lov\'asz swap with $\mu_{k,k-1} \neq 0$, the sum of squares $\sum_i p_i^2$ strictly decreases. Since the log-determinant $\sum_i p_i$ is conserved, so does the profile variance.
\end{corollary}

\begin{proof}
The function $\mathbf{x} \mapsto \sum x_i^2$ is strictly Schur-convex, so by Theorem~\ref{thm:variance} the claim follows. Concretely, writing $s = a + b$ (preserved),
\[
  p_{k-1}^2 + p_k^2
    = \tfrac{1}{2}\bigl(s^2 + (a-b)^2\bigr), \qquad
  (p'_{k-1})^2 + (p'_k)^2
    = \tfrac{1}{2}\bigl(s^2 + (a-b-2\varepsilon)^2\bigr).
\]
Since $\varepsilon \in (0,\, a-b)$, we have $(a-b-2\varepsilon)^2 < (a-b)^2$, and $\sum_i p_i^2$ strictly decreases.
\end{proof}

This variance decrease matters because $\sum p_i^2$ is minimized by the flat profile at fixed $\sum p_i$. Under the Lov\'asz gap constraints of Proposition~\ref{prop:gsa_minvar}, the minimizer becomes the GSA arithmetic progression. In that constrained $\ell^2$ sense, each swap pushes the profile toward GSA shape.

\begin{remark}[Entropy side]
\label{rem:entropy}
If the normalized profile $\mathbf{p}/\sum_j p_j$ has positive coordinates, then the same T-transform makes any \emph{strictly} Schur-concave functional of it increase at every Lov\'asz swap with $\mu_{k,k-1} \neq 0$; normalized Shannon entropy and R\'enyi entropies of positive finite order are instances. Strictness matters here. R\'enyi min-entropy $H_\infty$ is Schur-concave but not strictly so, so most swaps leave it unchanged. We do not use this direction in what follows.
\end{remark}

\begin{remark}
\label{rem:schur}
The unconstrained minimum of $\sum p_i^2$ at fixed $\sum p_i$ is the flat profile $p_i = \bar p$, which is LLL-reduced but not the algorithm's target; LLL stops at the first violation-free profile it encounters rather than at the global variance minimum. Proposition~\ref{prop:gsa_minvar} describes the opposite edge. The gap constraint $p_{i-1} - p_i \ge c_\delta$ selects profiles at or above the worst-case Lov\'asz threshold, and on that boundary the GSA profile of slope $-c_\delta$ is the unique minimum-variance shape. So the GSA slope follows variationally from $\delta$ alone, but the trajectory usually crosses that boundary and stops just inside the LLL stopping set. Variance monotonicity explains why the profile flattens toward the boundary, not why it should end on the GSA line. Appendix~\ref{app:universality} shows the experimental counterpart: normalized post-LLL profiles collapse to a highly consistent shape, with pairwise correlation above $0.91$ for all pairs with $d \ge 30$ and above $0.97$ for all pairs with $d \ge 60$, across dimensions up to $d = 140$.
\end{remark}

Two structural consequences follow from the T-transform picture. Variationally, the worst-case Lov\'asz boundary $c_\delta$ from~\eqref{eq:cdelta} pins the constrained minimum-variance profile to the GSA arithmetic progression. Dissipatively, the realized trajectory obeys an exact telescoping identity.

\begin{proposition}[GSA profile as constrained minimum-variance profile]
\label{prop:gsa_minvar}
Let $c_\delta = \tfrac{1}{2}\ln\!\bigl(\delta - \tfrac{1}{4}\bigr)^{-1} > 0$ be the worst-case log-gap from~\eqref{eq:cdelta}.
Among all profiles $\mathbf{p} \in \mathbb{R}^d$ satisfying
\begin{equation}\label{eq:lovasz_gap}
  \sum_{i=1}^d p_i = L \quad\text{and}\quad
  p_{i-1} - p_i \geq c_\delta \quad (i = 2,\ldots,d),
\end{equation}
the unique minimizer of $\sum_{i=1}^d p_i^2$ is the arithmetic progression
\begin{equation}\label{eq:gsa_profile}
  p_i^* = \frac{L}{d} + c_\delta\,\frac{d+1-2i}{2},
  \qquad i = 1,\ldots,d,
\end{equation}
which is the idealized GSA profile with slope $-c_\delta$.
\end{proposition}

\begin{proof}
Let $\mathbf p^*$ denote the arithmetic progression in~\eqref{eq:gsa_profile}. Direct computation shows $\sum_i p_i^* = L$ and $p_{i-1}^* - p_i^* = c_\delta$ for all $i$, so $\mathbf p^*$ is feasible. We show that any other feasible $\mathbf p$ has strictly larger objective.

Set $\mathbf q = \mathbf p - \mathbf p^*$. Feasibility of $\mathbf p$ is equivalent to
\begin{equation}\label{eq:gsa_minvar_q}
  \sum_{i=1}^d q_i = 0, \qquad q_1 \ge q_2 \ge \cdots \ge q_d,
\end{equation}
since the sum constraint becomes $\sum q_i = 0$ and the gap constraint becomes $q_{i-1} - q_i = (p_{i-1}-p_i) - c_\delta \ge 0$. Expanding,
\[
  \sum_{i=1}^d p_i^2 = \sum_{i=1}^d (p_i^*)^2 + 2\sum_{i=1}^d p_i^* q_i + \sum_{i=1}^d q_i^2.
\]
Using $p_i^* = L/d - c_\delta\bigl(i - (d+1)/2\bigr)$ and $\sum_i q_i = 0$, the constants drop out and
\[
  \sum_{i=1}^d p_i^* q_i = -c_\delta \sum_{i=1}^d i\, q_i.
\]
By Chebyshev's sum inequality, applied to the increasing sequence $(i)_{i=1}^d$ and the nonincreasing sequence $(q_i)_{i=1}^d$,
\[
  d \sum_{i=1}^d i\, q_i \;\le\; \Bigl(\sum_{i=1}^d i\Bigr)\Bigl(\sum_{i=1}^d q_i\Bigr) = 0,
\]
so $\sum_i p_i^* q_i \ge 0$. Therefore
\[
  \sum_{i=1}^d p_i^2 \;\ge\; \sum_{i=1}^d (p_i^*)^2 + \sum_{i=1}^d q_i^2 \;\ge\; \sum_{i=1}^d (p_i^*)^2,
\]
which proves the bound. Equality forces $\sum_i q_i^2 = 0$, hence $\mathbf q = \mathbf 0$ and $\mathbf p = \mathbf p^*$. The minimizer is unique.
\end{proof}

\begin{remark}[The GSA slab]
\label{rem:anchored_fail}
A score of the form $f(\mathbf p)=\sum_i h(p_i-p_i^*)$, anchored on the GSA target above, is tempting because it measures distance to the variational profile. The anchor breaks symmetry and fixes a position-dependent target. More importantly, $\mathbf p^*$ is a worst-case boundary. LLL usually stops on the flatter side, strictly inside the constraint polytope~\eqref{eq:lovasz_gap}. A valid Lov\'asz swap can close a local gap and still move the two affected coordinates farther from $\mathbf p^*$, so an anchored selector can reject a swap that LLL itself accepts. Appendix~\ref{app:variants} confirms the same pattern empirically for residual-targeted G-DLLL variants.
\end{remark}

\begin{proposition}[Variance dissipation identity]
\label{prop:convergence}
Let $\mathbf p(t)=(p_1(t),\ldots,p_d(t))$ be the GSO log-norm profile after $t$ Lov\'asz swaps, $t=0,\ldots,N$, and let $\mathbf p^*$ be the constrained GSA profile of Proposition~\ref{prop:gsa_minvar} for the conserved sum $L=\sum_i p_i(t)$. Define the relative variance
\begin{equation}\label{eq:excess_var}
  V(t) \;=\; \sum_{i=1}^d p_i(t)^2 \;-\; \sum_{i=1}^d (p_i^*)^2.
\end{equation}
At step $s\in\{1,\ldots,N\}$, acting at position $k_s$, write the absolute pre- and post-swap log-gaps
\[
  \Delta_s = \bigl|p_{k_s-1}(s{-}1) - p_{k_s}(s{-}1)\bigr|,\qquad
  \Delta'_s = \bigl|p_{k_s-1}(s) - p_{k_s}(s)\bigr|,
\]
and set $\varepsilon_s = \tfrac{1}{2}(\Delta_s - \Delta'_s)$. Then $\varepsilon_s \in [0,\Delta_s)$, with $\varepsilon_s=0$ exactly when the swap is degenerate ($\mu_{k_s,k_s-1}=0$, a pure transposition), and
\begin{equation}\label{eq:vdrop}
  V(s{-}1) - V(s) \;=\; \tfrac12\bigl(\Delta_s^2 - (\Delta'_s)^2\bigr) \;=\; 2\varepsilon_s(\Delta_s - \varepsilon_s) \;\ge\; 0.
\end{equation}
Telescoping gives
\begin{equation}\label{eq:convergence_identity}
  V(0) - V(N) \;=\; \sum_{s=1}^{N} 2\varepsilon_s(\Delta_s - \varepsilon_s),
\end{equation}
hence the swap-count lower bound
\[
  N \;\ge\; \frac{V(0)-V(N)}{\max_s 2\varepsilon_s(\Delta_s - \varepsilon_s)}.
\]
\end{proposition}

\begin{proof}
Since the reference sum $\sum_i (p_i^*)^2$ depends only on the conserved $L$, it is constant along the trajectory, so $V(s{-}1)-V(s) = \sum_i p_i(s{-}1)^2 - \sum_i p_i(s)^2$. At step $s$, only the pair at positions $k_s{-}1, k_s$ changes; write $a,b$ for the pre-swap log-norms and $a',b'$ for the post-swap log-norms. Every Lov\'asz swap preserves the local sum $a+b=a'+b'$ (the non-degenerate case is the T-transform of Theorem~\ref{thm:variance}; the degenerate case is a transposition). The identity $x^2+y^2=\tfrac12((x+y)^2+(x-y)^2)$ then gives
\[
  (a^2+b^2)-(a'^2+b'^2)
  \;=\; \tfrac12\bigl((a-b)^2 - (a'-b')^2\bigr)
  \;=\; \tfrac12\bigl(\Delta_s^2 - (\Delta'_s)^2\bigr),
\]
where the second equality uses $|a-b|=\Delta_s$ and $|a'-b'|=\Delta'_s$. Substituting $\Delta'_s = \Delta_s - 2\varepsilon_s$ yields $2\varepsilon_s(\Delta_s-\varepsilon_s)$. Theorem~\ref{thm:variance} gives $\Delta'_s\in[0,\Delta_s]$, with $\Delta'_s=\Delta_s$ only in the degenerate case, so $\varepsilon_s\in[0,\Delta_s)$ and $\varepsilon_s=0$ exactly there. Summing over $s$ gives~\eqref{eq:convergence_identity}, and $V(0)-V(N)\le N\max_s[V(s{-}1)-V(s)]$ gives the bound on $N$.
\end{proof}

% ===================================================================
\section{Algorithmic Consequences of Majorization}
\label{sec:algorithms}
% ===================================================================

For adjacent swaps, Corollary~\ref{cor:variance} shows that LLL decreases $\sum p_i^2$ one local move at a time. Deep insertions are less rigid, since a long move need not inherit the same majorization order. Even so, the adjacent-swap picture suggests two guides for nonlocal moves. One is to follow Schur-convex scores built from the log-profile. The other is to ask whether variance itself remains a useful descent principle.

\subsection{Deep insertions and the temperature spectrum}
\label{sec:deep}

A deep insertion rearranges the whole window between positions~$j$ and~$k$ in a single move, which makes the choice of objective inside the window a real design question. We score candidates with the standard deep-insertion recursion of Section~\ref{sec:prelim}~\cite{yamaguchi2017explicit}.

The adjacent-swap majorization law does not extend verbatim to deep insertions, because a long move can contain intermediate transpositions that widen a local gap, so a Lov\'asz-admissible insertion need not satisfy $\mathbf{p}' \prec \mathbf{p}$. The selectors below therefore enforce descent on their chosen score candidate by candidate, and termination still follows from the LLL potential~$\Phi$ of Section~\ref{sec:prelim}~\cite{schnorr1994lattice,fontein2014potlll}.

\subsubsection{Variance-greedy selector}

Our first selector is the direct variance-greedy rule, which we call \emph{Deep-Var}. At each iteration, it enumerates all admissible pairs $(k, j)$ satisfying the generalized Lov\'asz condition (Definition~\ref{def:gen_lovasz}), evaluates the post-insertion window through the standard cascade, and chooses the pair that maximally decreases $\sum p_i^2$. If no candidate gives positive descent, the algorithm stops. This requires $O(d^2)$ candidates per iteration, each needing $O(k-j)$ work to evaluate the window, hence $O(d^3)$ work per iteration.

\subsubsection{Thermal family}

Deep-Var and SS-GG are endpoints of a one-parameter family of greedy objectives. Define the thermal potential
\begin{equation}
\label{eq:thermal}
\phi_\alpha(\mathbf{r}) \;=\; \sum_{i=1}^{d} r_i^{\,\alpha}
  \;=\; \sum_{i=1}^{d} e^{\,2\alpha\, p_i}\,,
  \qquad \alpha > 0\,,
\end{equation}
with $r_i = \lVert\mathbf{b}_i^*\rVert^2$ and $r_i = e^{2p_i}$. Since $x \mapsto e^{2\alpha x}$ is convex for every $\alpha > 0$, $\phi_\alpha$ is a Schur-convex function of the log-profile $\mathbf{p}$. For non-degenerate adjacent swaps, where $\mathbf{p}' \prec \mathbf{p}$, $\phi_\alpha$ decreases strictly; degenerate swaps leave it unchanged. For deep insertions we define the \emph{thermal selector} to score each candidate $(k,j)$ by
\begin{equation}
\label{eq:thermal_mono}
\Delta\phi_\alpha(k,j)
  \;=\; \phi_\alpha(\mathbf{r}) - \phi_\alpha(\mathbf{r}')
\end{equation}
and pick the pair that maximizes $\Delta\phi_\alpha$, accepting only candidates with $\Delta\phi_\alpha > 0$. At $\alpha = 1$ this recovers SS-GG ($\phi_1 = \sum r_i$). A Taylor expansion in $\alpha$ around $\alpha = 0$ shows that the selector approaches Deep-Var in the limit $\alpha \to 0$:
$\Delta\phi_\alpha = 2\alpha^2\,\Delta(\sum p_i^2) + O(\alpha^3)$,
since the first-order term $2\alpha\,\Delta(\sum p_i)$ vanishes (sublattice determinant invariance), so for small $\alpha$ the thermal family ranks candidates identically to variance descent. As $\alpha \to \infty$, $\phi_\alpha^{1/\alpha} \to \max_i r_i$, a max-norm selector.

Each fixed value of $\alpha$ defines a concrete greedy selector with the same candidate loop as Deep-Var; only the score changes. In Section~\ref{sec:initadaptive} we choose $\alpha$ from the initial profile and call the resulting fixed-$\alpha$ rule \emph{Thermal-Adaptive}.

\subsection{Compatible score families}
\label{sec:canonicity}

The thermal family is one way to use the swap geometry. A natural question is which other scalar objectives are guaranteed to see every Lov\'asz swap as progress. Pot-DeepLLL~\cite{fontein2014potlll}, SS-DeepLLL~\cite{yasuda2017analysis,yasuda2019new}, and the X-GG framework~\cite{bhattacherjee2025greedy} each established this property by a direct calculation tailored to their own potential. The T-transform reading of an adjacent swap collapses those separate arguments into one. A non-degenerate Lov\'asz swap smooths a local kink in the log-profile, so any functional that strictly decreases under such smoothing tracks the LLL mechanism by construction. We call such scores \emph{Lov\'asz-compatible} and show that two sufficient classes, derived from majorization and from a weighted convexity argument, together cover the standard selectors.

\begin{definition}[Lov\'asz-compatible score]
\label{def:lovasz_compatible}
For $\mathbf{p} \in \mathbb{R}^d$, $k \in \{2,\ldots,d\}$, and $\varepsilon \in (0,\,p_{k-1}-p_k)$, the adjacent T-transform $T_{k,\varepsilon} = T_{k-1,k,\varepsilon}$ acts by
\[
  T_{k,\varepsilon}\mathbf{p}
  = (p_1,\ldots,p_{k-2},\,p_{k-1}-\varepsilon,\,p_k+\varepsilon,\,p_{k+1},\ldots,p_d).
\]
A score $f \colon \mathbb{R}^d \to \mathbb{R}$ is \emph{Lov\'asz-compatible} if, for every profile $\mathbf{p}$ and every position $k$ with $p_{k-1} > p_k$,
\begin{equation}\label{eq:lovasz_compatible}
  f(\mathbf{p}) \;>\; f\bigl(T_{k,\varepsilon}\mathbf{p}\bigr)
  \qquad \forall\,\varepsilon \in (0,\,p_{k-1} - p_k).
\end{equation}
\end{definition}

For a $C^1$ score, differentiating the path $g(\varepsilon) = f\bigl(T_{k,\varepsilon}\mathbf p\bigr)$ at $\varepsilon = 0$ shows that~\eqref{eq:lovasz_compatible} forces $\partial_{k-1} f(\mathbf p) \ge \partial_k f(\mathbf p)$ whenever $p_{k-1} > p_k$. The two classes below ensure strict descent along the entire T-transform path, without assuming differentiability. The first invokes majorization directly; the second uses a direct two-coordinate calculation.

\begin{proposition}[Two families of Lov\'asz-compatible scores]
\label{prop:canonical}
Each score in either class below is Lov\'{a}sz-compatible:
\begin{enumerate}
  \item[(i)] $f$ is symmetric and strictly Schur-convex~\cite{marshall1979inequalities}.
  \item[(ii)] $f(\mathbf{p}) = \sum_i w_i \psi(p_i)$, with weights $w_1 > w_2 > \cdots > w_d > 0$ ordered with the profile, and $\psi$ a convex strictly increasing function.
\end{enumerate}
The strict weight ordering is essential when $\psi$ is affine. If the weights flatten to equality, class~(ii) remains Lov\'asz-compatible only in the strictly convex case, where it becomes a symmetric separable instance of class~(i).
\end{proposition}

\begin{proof}
Both classes reduce to the active pair. Set $a = p_{k-1} > b = p_k$ and let $\varepsilon \in (0,\, a-b)$.

\smallskip\noindent\emph{(i)} A T-transform produces $T_{k,\varepsilon}\mathbf p \prec \mathbf p$ with strict majorization for $\varepsilon \in (0,\,a-b)$, so strict Schur-convexity gives
\[
  f(T_{k,\varepsilon}\mathbf p) < f(\mathbf p).
\]
That is exactly Lov\'asz-compatibility.

\smallskip\noindent\emph{(ii)} The score drop splits into an equal-weight smoothing term and an extra reward for moving mass away from the earlier, more heavily weighted coordinate. Writing $w_{k-1} = w_k + (w_{k-1}-w_k)$,
\[
  f(\mathbf p) - f(T_{k,\varepsilon}\mathbf p)
  = w_k\bigl[\psi(a)+\psi(b) - \psi(a{-}\varepsilon) - \psi(b{+}\varepsilon)\bigr]
  + (w_{k-1}-w_k)\bigl[\psi(a) - \psi(a{-}\varepsilon)\bigr].
\]
The first bracket is nonnegative because convexity of $\psi$ makes $\sum \psi(p_i)$ Schur-convex, so equal-weight smoothing cannot increase it. The second bracket is strictly positive because $w_{k-1} > w_k$ and $\psi$ is strictly increasing. Their sum is therefore strictly positive, which is~\eqref{eq:lovasz_compatible}. Pot-DeepLLL's affine choice $\psi(p) = 2p$ shows why strict monotonicity of $\psi$ already controls the second term; convexity is only needed for the first.

The final statement follows from the same decomposition. When the weights are equal, the second term vanishes; strict descent must then come from strict convexity in the first term. For affine $\psi$, the equal-weight score is just a multiple of $\sum_i p_i$ and is conserved by a T-transform.
\end{proof}

\begin{corollary}[Standard separable instances]
\label{cor:standard_scores}
If $\psi$ is strictly convex, then $f(\mathbf p)=\sum_i \psi(p_i)$ is Lov\'asz-compatible. In particular, the thermal family $\phi_\alpha$~\eqref{eq:thermal} is Lov\'asz-compatible. The $r$-gradient is $\partial_{r_i}\phi_\alpha=\alpha r_i^{\alpha-1}$; it is independent of both position and magnitude exactly at $\alpha=1$, where $\phi_1=\sum_i r_i$ is SS-GG.

If $w_1>\cdots>w_d>0$ and $\psi$ is convex and strictly increasing, then $\sum_i w_i\psi(p_i)$ is Lov\'asz-compatible. Pot-DeepLLL~\cite{fontein2014potlll} is the affine instance $\psi(p)=2p$, $w_i=d-i+1$, and the mixed family $f_{\alpha,\beta}(\mathbf r)=\sum_i(d-i+1)^\beta r_i^\alpha$ is the exponential instance $\psi(p)=e^{2\alpha p}$, $w_i=(d-i+1)^\beta$.
\end{corollary}

\begin{proof}
For $f=\sum_i\psi(p_i)$, the score is symmetric. A symmetric separable sum with strictly convex $\psi$ is strictly Schur-convex~\cite{marshall1979inequalities}, so Proposition~\ref{prop:canonical}(i) applies. The thermal family is the choice $\psi(p)=e^{2\alpha p}$, which is strictly convex for $\alpha>0$, and the statement about $\alpha=1$ follows by differentiating in the variables $r_i$.

The weighted statement is exactly Proposition~\ref{prop:canonical}(ii), with the listed choices of $\psi$ and $w_i$. For Pot-DeepLLL, $\mathrm{Pot}(B)=\prod_i r_i^{d-i+1}$, and taking logarithms gives $\ln\mathrm{Pot}(B)=\sum_i 2(d-i+1)p_i = 2\Phi$, with $\Phi$ the LLL potential of Section~\ref{sec:prelim}, so descent on $\mathrm{Pot}$ is equivalent to descent on $f$.
\end{proof}

Thermal-Adaptive sits in the symmetric Schur-convex class, Pot-DeepLLL in the weighted separable class, and SS-GG is the point $\alpha=1$ of the thermal family, where the score is linear in the squared GSO norms and has constant $r$-gradient.

\begin{remark}[Two-position update]
\label{rem:ssgg_inc}
The cascade has one useful implementation consequence. As the candidate insertion point $j$ moves downward, only two post-insertion entries change, so any separable score can be updated incrementally. At $\alpha=1$, the SS-GG score simplifies further. Using $P_\ell-P_{\ell+1}=\mu_{k,\ell}^2 r_\ell$, the score drop satisfies
\begin{equation}
\label{eq:ssgg_inc}
\Delta\phi_1(k,j) \;=\; \sum_{\ell=j}^{k-1} \mu_{k,\ell}^{\,2}\, r_\ell \left( \frac{r_\ell}{P_\ell} - 1 \right),
\end{equation}
an identity already used in the SS-DeepLLL reference implementation~\cite{yasuda2019new,bhattacherjee2025greedy}. Thus SS-GG is not only the point of the thermal family at which the gradient is independent of position and magnitude; it is also the cheapest member to score in the standard descending-$j$ loop. For $\alpha\ne1$, the same two-entry update remains simple but no longer collapses to this recurrence.
\end{remark}

\subsection{Thermal-Adaptive selector}
\label{sec:initadaptive}

SS-GG~\cite{bhattacherjee2025greedy} corresponds to $\alpha = 1$ in the thermal family and is the distinguished point of Corollary~\ref{cor:standard_scores}: its $r$-gradient is independent of position and magnitude, so every admissible candidate is scored only by how much it changes $\sum r_i$, with no per-position weighting. Thermal-Adaptive stays within the same strictly Schur-convex family but lets the initial profile choose $\alpha$, then keeps that value fixed throughout reduction. On $q$-ary inputs it lands at $\alpha = 1$ and recovers SS-GG exactly. On flatter profiles it shifts to $\alpha > 1$ and reduces the operation count by 8 to 15\% in our benchmarks (Section~\ref{sec:experiments}).

The mechanism is easiest to see through the per-entry gradient $\partial\phi_\alpha / \partial r_i = \alpha\, r_i^{\,\alpha-1}$. At $\alpha = 1$ this is constant, so all norms contribute equally to the score. That suits heterogeneous profiles, where the largest $r_i$ already dominate changes in $\sum r_i$. On homogeneous profiles the $r_i$ are similar, many candidate insertions produce nearly the same $\Delta(\sum r_i)$, and the selector becomes nearly blind. Raising $\alpha$ above~1 makes the gradient grow with~$r_i$, rewards insertions that shrink the largest norms, and breaks those near-ties. On $q$-ary lattices the opposite regime applies: the bimodal initial profile is exaggerated by large $\alpha$ until the selector becomes myopic, with operation counts growing roughly two orders of magnitude as $\alpha$ moves from 1 to 5 (Figure~\ref{fig:thermal_spectrum}a).

To turn this picture into a concrete rule, we need a single number that says how flat the input profile is. Write
\begin{equation}
  \mu_0 \;=\; \frac{1}{d}\sum_{i=1}^{d} \ln r_{0,i},
  \qquad
  \sigma_0 \;=\; \sqrt{\tfrac{1}{d}\sum_{i=1}^{d} \bigl(\ln r_{0,i} - \mu_0\bigr)^{2}},
  \qquad
  \mathrm{CV}_0 \;=\; \frac{\sigma_0}{|\mu_0|},
\end{equation}
where $\mathbf{r}_0 \in \mathbb{R}^d_{>0}$ is the initial GSO norm-squared vector and the standard deviation is taken in log-space. The schedule itself is a power-law that maps a flat profile (small $\mathrm{CV}_0$) to a hot $\alpha$ and a heterogeneous profile (large $\mathrm{CV}_0$) to a cold $\alpha$:
\begin{equation}
\label{eq:adaptive_alpha}
  \alpha_0
  \;=\; \max\!\Bigl(\alpha_{\min},\;
        \Bigl(\frac{2}{1 + \mathrm{CV}_0}\Bigr)^{\!\gamma}\Bigr).
\end{equation}

Why $\mathrm{CV}_0$ rather than the obvious alternative? The log of the ratio between the largest and smallest per-entry gradient is exactly $(\alpha-1)(\max_i \ln r_i - \min_i \ln r_i)$, so the natural measure of profile spread is the log-space range. We use $\sigma_0/|\mu_0|$ instead because it is a smoother surrogate for the same quantity: the standard deviation responds to the bulk of the profile rather than to the worst pair, and dividing by $|\mu_0|$ removes the dependence on overall scale. Both choices order our benchmark profiles in the same way; $\mathrm{CV}_0$ has the practical advantage that every $q$-ary profile we tested gives $\mathrm{CV}_0 \approx 1.0$, which makes the next anchor condition exact rather than approximate.

The normalization is calibrated against the idealized $q$-ary picture. Half of the $r_i$ sit near $q^2$ and half near $1$, so $\ln r_i$ is approximately a two-point distribution with masses at $2\ln q$ and $0$. Its mean is $\ln q$ and its standard deviation is also $\ln q$, giving $\mathrm{CV}_0 = 1$. Plugging this into~\eqref{eq:adaptive_alpha} gives $\alpha_0(\mathrm{CV}_0{=}1) = 1$, the \emph{$q$-ary anchor}, and that is what makes Thermal-Adaptive recover SS-GG exactly on $q$-ary lattices.

The exponent $\gamma$ then controls how steeply $\alpha_0$ rises as the profile flattens. The fixed-$\alpha$ sweep (Figure~\ref{fig:thermal_spectrum}a) shows that cost curves are flat for homogeneous profiles once $\alpha \gtrsim 2$, so $\gamma$ only needs to place $\alpha_0$ above this plateau. We use $\gamma = 2$, which gives $\alpha_0 \approx 2.9$ at $\mathrm{CV}_0 \approx 0.17$, and the benchmark improvement is stable across $\gamma \in [1.5, 3]$. A fixed $\alpha > 1$ cannot replace adaptation: $\alpha = 3$ improves 5 to 8\% on flat profiles but degrades 24 to 28\% on $q$-ary (Figure~\ref{fig:thermal_spectrum}a).

The floor $\alpha_{\min} = 0.4$ guards against degenerate inputs where $|\mu_0| \to 0$ blows up $\mathrm{CV}_0$ and pushes the formula below~$1$. On every lattice in our benchmark $|\mu_0|$ stays well away from zero and the $\max$ is never active, so the value of $\alpha_{\min}$ in $(0, 1)$ does not change any reported number.

The parameter $\alpha$ is computed once and held fixed. Continuous re-adaptation from the evolving profile is counterproductive. As reduction smooths the profile, $\mathrm{CV}$ drops, pushing $\alpha$ above 1 on $q$-ary lattices and into the high-cost regime, increasing operation counts by 15 to 35\% relative to SS-GG. Appendix~\ref{app:variants} reports a periodic re-estimation variant (Thermal-Sched); its $(W, \delta_0)$ trade-off is nearly flat in the re-estimation period, which rules out cadence as a useful knob.

For every $\alpha > 0$, $\phi_\alpha$ is Schur-convex, so the selector enforces $\Delta\phi_\alpha > 0$ and termination follows from $\Phi$ as before. The per-candidate cost matches SS-GG: one cascade evaluation plus a single power-sum over $k - j + 1$ entries.

\subsection{Geodesic Deep-LLL (G-DLLL)}
\label{sec:gdlll}

Deep-Var maximizes $\Delta V(k,j)$ per insertion, targeting insertion count. A deep insertion, however, still pays for the cascaded GSO update across the window, whose length is $k-j$ in equivalent adjacent swaps. G-DLLL lowers the equivalent-swap count $W$ even though it uses many more insertions $N$. Minimizing $W$ instead of $N$ calls for a different objective. We treat the resulting selector as a theory-side consequence of the variance-dissipation picture rather than as the paper's main practical algorithm.

\begin{proposition}[ROI lower bound on equivalent work]
\label{prop:roi}
Define the \emph{ROI efficiency} $\eta(k,j) = \Delta V(k,j)/(k-j)$ for $k > j$, the variance dissipation per unit of cascade update work.
Let $W = \sum_s (k_s - j_s)$ be the total number of equivalent adjacent swaps
over all $T$ insertions, and let
$\Delta V_{\mathrm{tot}} = \sum_i p_i(0)^2 - \sum_i p_i(T)^2$
be the total profile energy dissipated.  Then
\begin{equation}
\label{eq:roi_lb}
  W \;\geq\; \frac{\Delta V_{\mathrm{tot}}}{\displaystyle\max_{k,j}\,\eta(k,j)}.
\end{equation}
\end{proposition}
\begin{proof}
Each insertion decreases $\sum p_i^2$ by $\Delta V_s > 0$, so the drops telescope to $\Delta V_{\mathrm{tot}} = \sum_s \Delta V_s$. For insertion $s$, the definition of $\eta$ gives $\Delta V_s = (k_s-j_s)\,\eta(k_s,j_s)$, which separates the insertion depth from the per-unit energy gain. Bounding each efficiency by the global maximum and collecting the depths as $\sum_s(k_s-j_s)=W$ gives
\[
  \Delta V_{\mathrm{tot}}
  = \sum_s (k_s-j_s)\,\eta(k_s,j_s)
  \;\leq\; W \cdot \max_{k,j}\eta(k,j).
\]
Rearranging gives~\eqref{eq:roi_lb}.
\end{proof}

The bound~\eqref{eq:roi_lb} tightens when $\eta$ is large, which suggests an analogue of Deep-Var that maximizes efficiency rather than absolute descent. We call the resulting selector \emph{Geodesic Deep-LLL (G-DLLL)}. At each iteration, it enumerates the admissible pairs $(k,j)$ satisfying the generalized Lov\'asz condition (Definition~\ref{def:gen_lovasz}), scores each by the ROI efficiency $\eta(k,j)$ of Proposition~\ref{prop:roi}, and chooses the pair maximizing $\eta$. If no candidate gives positive descent, the algorithm stops. Termination follows from the LLL potential~$\Phi$, which decreases strictly at each accepted insertion. When all candidates have the same depth $k-j$, the score $\eta$ and $\Delta V$ rank pairs identically, so G-DLLL reduces to Deep-Var.

\paragraph{Candidate shortlist.}
Enumerating all $\binom{d}{2}$ pairs at every step would scale poorly. Two observations cut the search without affecting the rule. First, the local deficit $(p_i - p_{i+1}) - c_\delta$~\eqref{eq:cdelta} is large only near the top of the profile, where GSO norms are most inflated relative to the target slope; on all tested lattice families, the top $K = \lceil d/3 \rceil$ source positions by descending deficit account for over 95\% of the insertable $\Delta V$ budget, and halving or doubling $K$ leaves the output unchanged. Second, an insertion that barely moves the profile still pays a full cascade update, so we require $\Delta V(k,j) > \tau \sum_i p_i^2$ with $\tau = 0.01$. The threshold tightens automatically as the profile flattens, relaxes on disordered inputs, and is empirically flat in $\tau \in [0.005, 0.02]$. Both filters are search heuristics; they do not alter the selection rule of Proposition~\ref{prop:roi}.

% ===================================================================
\section{Experimental Evaluation}
\label{sec:experiments}
% ===================================================================

We use $\delta = 0.99$ throughout. The main deep-insertion benchmarks use three lattice families spanning the range from homogeneous to strongly structured profiles: \emph{Gaussian} (i.i.d.\ entries $\sim \mathcal{N}(0, 25)$, $d \in \{20,30,\ldots,200\}$), \emph{$q$-ary} ($q = 1009$, $k = d/2$; the standard Hermite-normal-form construction $[q\,I_k, 0; A, I_{d-k}]$ due to Ajtai~\cite{ajtai1996generating} yields roughly $d/2$ Gram-Schmidt norms near~$q$ and $d/2$ near~$1$), and \emph{Goldstein-Mayer}~\cite{goldstein2003equidistribution} ($k = 1$, $q$ a random $10d$-bit prime, following the same convention as~\cite{chen2011bkz,bhattacherjee2025greedy}). The fixed-temperature sweep in Figure~\ref{fig:thermal_spectrum} also includes a \emph{Uniform} family with i.i.d.\ entries in $\{-10,\ldots,10\}$. Each benchmark cell uses $n = 30$ independent lattices. We report operation count, wall-clock time, root-Hermite factor~$\delta_0$, and final profile variance $\sum p_i^2$.

We use two implementation layers. Per-swap tracing and other analyses that need the internal swap trajectory run in instrumented Python/\texttt{fpylll} code. Performance benchmarking of deep insertions runs in the C++/\texttt{fplll} code released as \texttt{variationaLLL} (Section~\ref{sec:data_avail}), which we treat as the canonical implementation. The experiments test two points, the value of adaptive Schur-convex scoring and the separation between selector quality and update cost.

\paragraph{Deep-insertion selectors.}
We compare three Schur-convex deep-insertion selectors against the standard LLL baseline: Deep-Var, SS-GG~\cite{bhattacherjee2025greedy}, and Thermal-Adaptive. Pot-GG is Lov\'asz-compatible by Proposition~\ref{prop:canonical}(ii), but its benchmark against SS-GG was already settled in~\cite{bhattacherjee2025greedy}, where SS-GG was shown to achieve uniformly better output quality across all tested dimensions; rerunning that comparison here would add nothing. Plain DeepLLL is omitted for the same reason. BKZ is not a direct baseline for this question, because it changes both the search primitive and the tuning axis at once, while our aim is to compare selector rules within the LLL-style deep-insertion family. Our SS-GG implementation uses~\eqref{eq:ssgg_inc} throughout the admissibility loop, including the initializer at $\ell = k{-}1$; using~\eqref{eq:ssgg_inc} uniformly avoids a scaling mismatch we observed between the first term ($\ell = k{-}1$) and the loop body in the reference implementation of~\cite{bhattacherjee2025greedy}, reported separately as a patch upstream.\footnote{GitHub pull request \#1 against the Greedy-Global-LLL repository: \url{https://github.com/GG-LLL/Greedy-Global-LLL/pull/1}. The patch preserves the sign of each $\Delta_{\mathrm{SS}}$ term so convergence is unaffected, but changes the $(k, j)$ pair selected when the profile is steep.}

Table~\ref{tab:deep} reports representative dimensions, while Figures~\ref{fig:deep_gaussian},~\ref{fig:deep_qary},~\ref{fig:deep_goldstein}, and~\ref{fig:thermal_spectrum} show the main trends.

\begin{table}[t]
\centering
\caption{Deep-insertion selectors: mean operation count $N$ and mean equivalent-swap count $W$ over $n=30$ independent lattices per cell, with standard errors in parentheses ($\delta = 0.99$, C++ implementation). $\Delta_{\mathrm{SS}}$ is Thermal-Adaptive's percentage reduction in $N$ relative to SS-GG; the $d{=}40$ Goldstein-Mayer difference is within one standard error of zero.}
\label{tab:deep}
\small
\setlength{\tabcolsep}{3pt}
\resizebox{\textwidth}{!}{%
\begin{tabular}{llrrrrrrr}
\toprule
Family & $d$ & Deep-Var $N$ & SS-GG $N$ & Therm.\ $N$ & $\Delta_{\mathrm{SS}}$ & Deep-Var $W$ & SS-GG $W$ & Therm.\ $W$ \\
\midrule
Gaussian &  40 &   77 (3) &   71 (2) &   62 (2) & $+$13\% &   743 (24) &   736 (22) &   691 (19) \\
         &  80 &  201 (6) &  178 (5) &  158 (4) & $+$11\% &  2334 (47) &  2360 (43) &  2348 (42) \\
         & 120 &  348 (8) &  307 (8) &  260 (6) & $+$16\% &  5061 (75) &  5070 (67) &  5028 (56) \\
         & 160 &  521 (11) &  465 (9) &  404 (9) & $+$13\% &  9155 (110) &  9232 (99) &  9139 (118) \\
\addlinespace
$q$-ary  &  40 &  160 (3) &  155 (3) &  155 (3) & $0\%$   &  3432 (26) &  3298 (29) &  3298 (29) \\
         &  80 &  518 (11) &  626 (13) &  627 (13) & $0\%$   & 18100 (99) & 20561 (133) & 20572 (134) \\
         & 120 & 1074 (21) & 1927 (33) & 1927 (33) & $0\%$   & 46820 (199) & 66993 (399) & 66993 (399) \\
         & 160 & 1877 (46) & 5454 (106) & 5454 (106) & $0\%$   & 86423 (472) &169756 (1216) &169756 (1216) \\
\addlinespace
Goldstein-Mayer
         &  40 &   61 (3) &   61 (3) &   64 (3) & $-$5\%  &   504 (36) &   572 (38) &   681 (40) \\
         &  80 &  310 (11) &  393 (11) &  322 (10) & $+$18\% &  5331 (136) &  7791 (143) &  6168 (144) \\
         & 120 &  781 (27) & 1421 (40) &  864 (20) & $+$39\% & 20524 (352) & 36790 (476) & 25007 (295) \\
         & 160 & 1660 (39) & 4851 (88) & 1956 (46) & $+$60\% & 51221 (567) &125042 (1246) & 68306 (666) \\
\bottomrule
\end{tabular}
}% end resizebox
\setlength{\tabcolsep}{6pt}
\end{table}

All three deep-insertion selectors reduce operation counts sharply relative to standard LLL. On Gaussian lattices (Figure~\ref{fig:deep_gaussian}), Thermal-Adaptive has the lowest operation count throughout the measured range; at $d \in \{40,80,120\}$ it uses 13, 11, and 16\% fewer operations than SS-GG, with essentially the same equivalent-swap count~$W$ and slightly better output quality. On a flat profile the squared norms $r_i$ all sit close to a common value, so $\Delta(\sum r_i)$ is nearly the same for many admissible candidates and SS-GG ranks them by tiny differences. Raising $\alpha$ above 1 weights large $r_i$ more strongly, which restores a clear ordering among those near-ties and steers the selector to candidates that shrink the largest entries.

On $q$-ary lattices (Figure~\ref{fig:deep_qary}), Thermal-Adaptive collapses back to SS-GG by construction. The contrast there is with Deep-Var. Deep-Var scores by $\sum p_i^2$, treating the two $q$-ary modes (entries near $\tfrac{1}{2}\ln q$ and entries near $0$) symmetrically in log-space, so it favors moves that close the log-space gap between the modes. The Lov\'asz stopping set is reached after fewer such moves than SS-GG needs under its linear-$r$ objective, and the basis it returns has worse output quality. At $d=120$, Deep-Var uses 1074 insertions against 1927 for SS-GG and Thermal-Adaptive, but $\delta_0 = 1.01537$ versus $1.01378$.

Goldstein-Mayer lattices (Figure~\ref{fig:deep_goldstein}) are not a duplicate of the $q$-ary case. The half-rank $q$-ary construction has a sharp bimodal profile by design, and Thermal-Adaptive's calibration sets $\alpha = 1$ on it, so on $q$-ary the rule recovers SS-GG. Goldstein-Mayer profiles are intermediate: they have structure inherited from the random prime $q$ but no clean two-plateau split. There Thermal-Adaptive picks $\alpha > 1$ and is slightly worse at $d=40$, then overtakes SS-GG at $d \in \{80, 120, 160\}$, where its operation-count advantage grows to 18, 39, and 60\%, with $W$ reduced by 21 to 45\% and wall-clock time by 8 to 47\%. SS-GG retains a small $\delta_0$ edge on this family (about $10^{-3}$) but at roughly double the compute cost. Across the three families, the theory does not single out one best deep score; it isolates a useful objective family and a concrete profile-adaptive path inside it.

\begin{figure}[t]
\centering
\includegraphics[width=\textwidth]{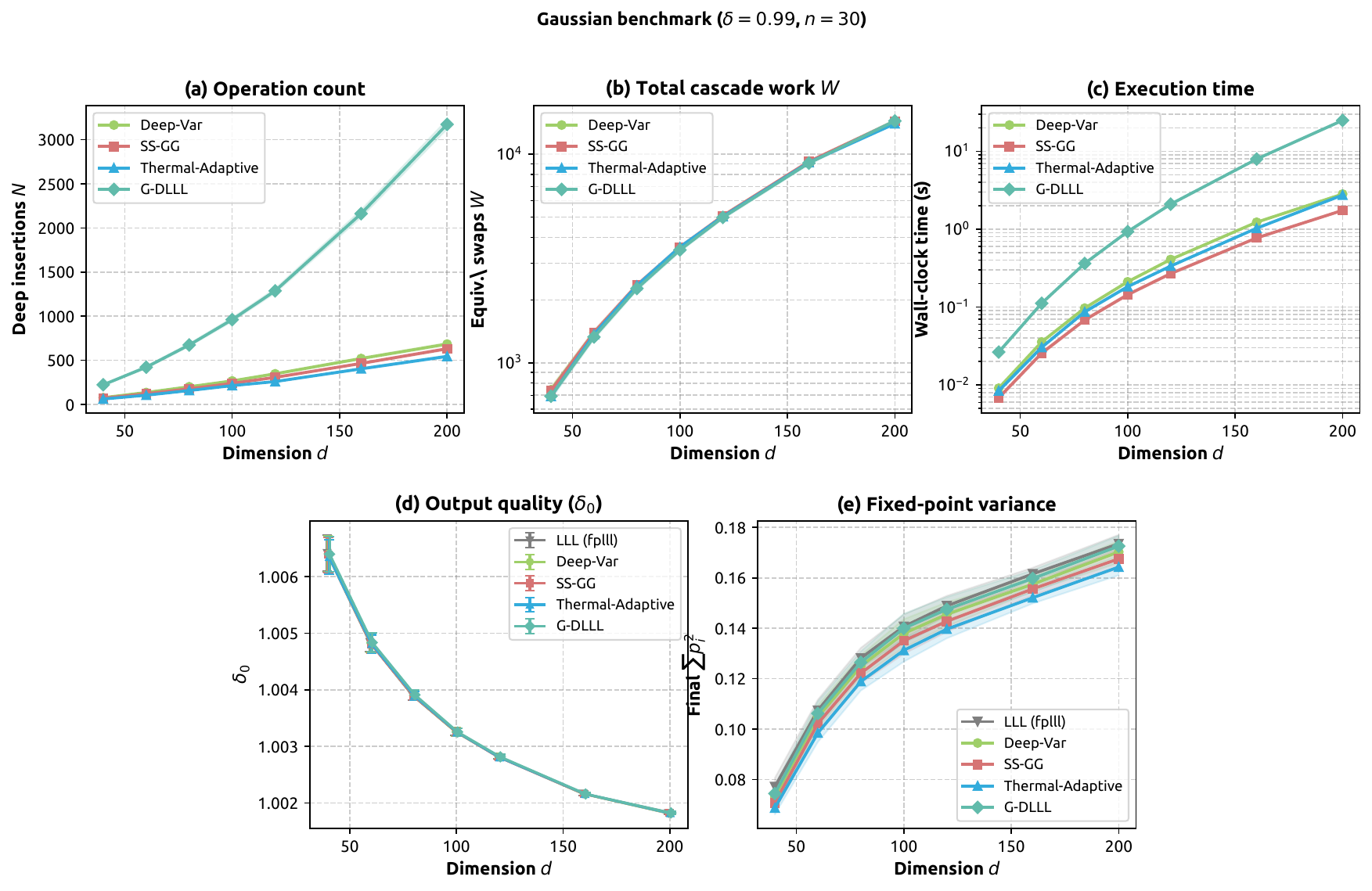}
\caption{Deep-insertion selectors on Gaussian lattices ($\mathcal{N}(0,25)$, $n=30$, $\delta=0.99$, $d$ up to 120, C++ implementation). Shaded bands show $\pm 1$ standard error. Panels show operation count $N$, equivalent-swap count $W$, total selector time, root-Hermite factor $\delta_0$, and final profile variance. Thermal-Adaptive uses 5 to 16\% fewer operations than SS-GG at comparable output quality.}
\label{fig:deep_gaussian}
\end{figure}

\begin{figure}[t]
\centering
\includegraphics[width=\textwidth]{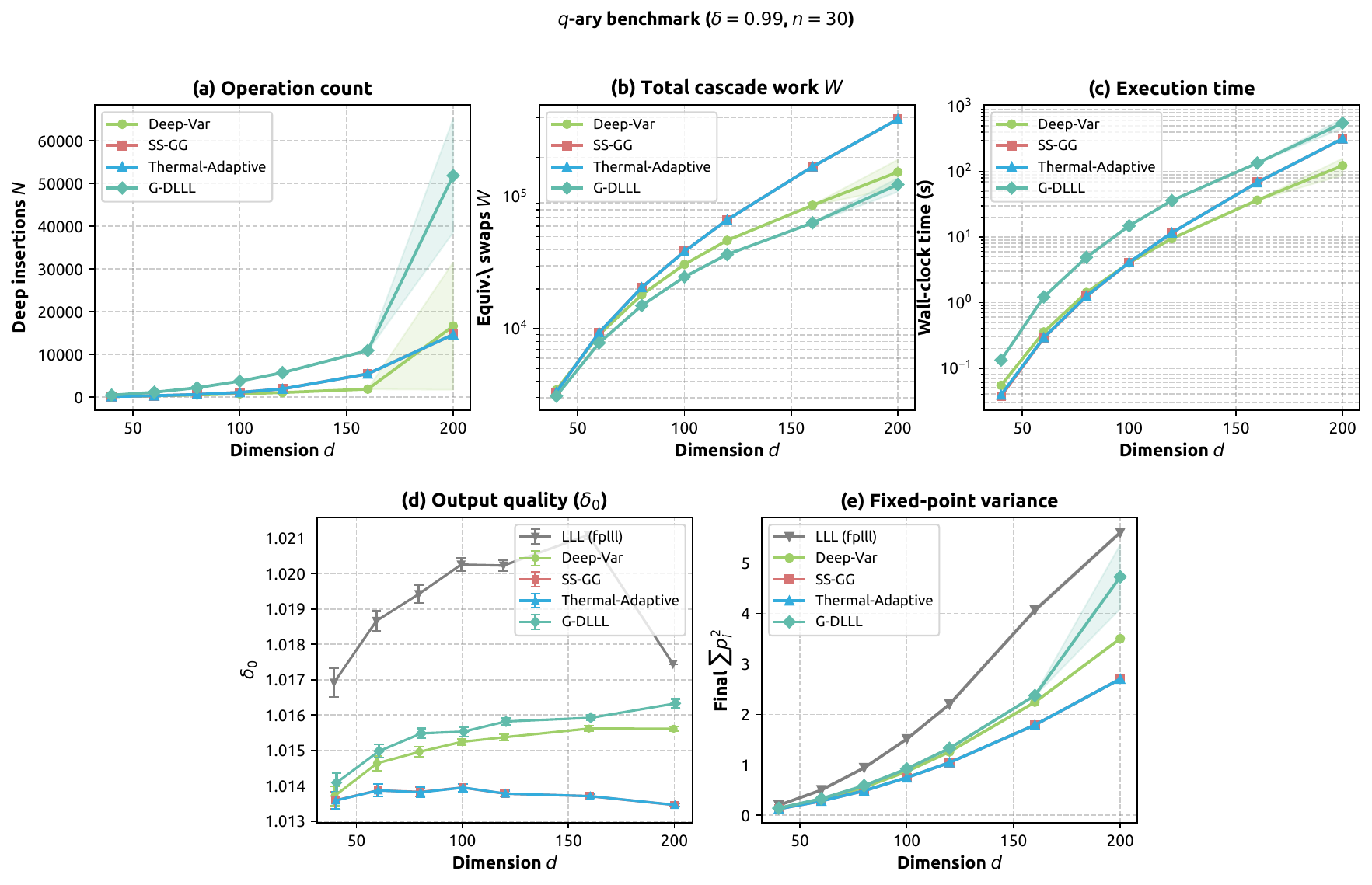}
\caption{Deep-insertion selectors on $q$-ary lattices ($q = 1009$, $k=d/2$, $n=30$, $\delta=0.99$, $d$ up to 120, C++ implementation). Shaded bands show $\pm 1$ standard error. Panels show operation count $N$, equivalent-swap count $W$, total selector time, root-Hermite factor $\delta_0$, and final profile variance. Thermal-Adaptive recovers SS-GG by construction, while Deep-Var uses fewer operations for $d \ge 80$ at slightly worse $\delta_0$.}
\label{fig:deep_qary}
\end{figure}

\begin{figure}[t]
\centering
\includegraphics[width=\textwidth]{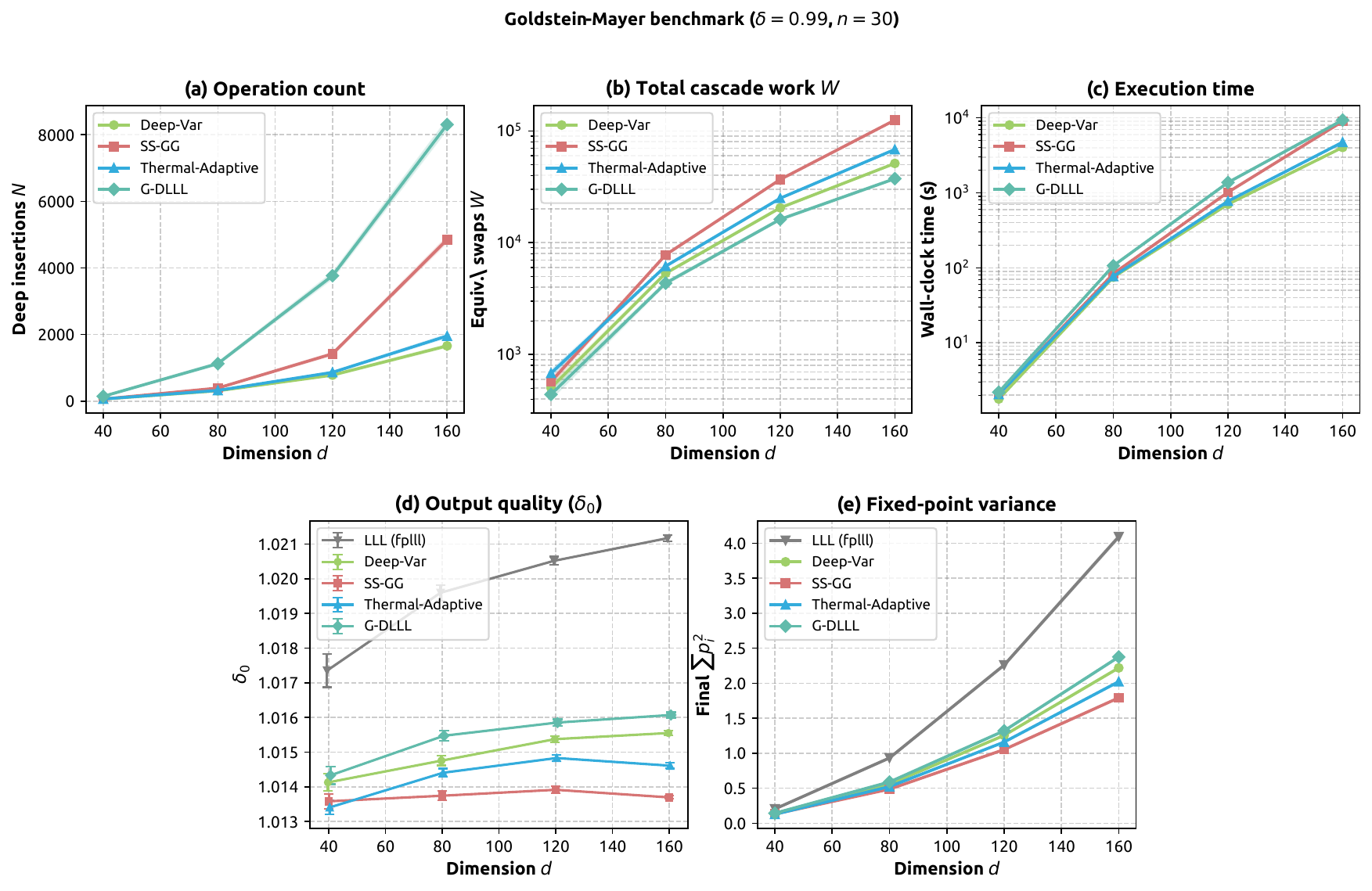}
\caption{Deep-insertion selectors on Goldstein-Mayer lattices ($k=1$, $q$ a random $10d$-bit prime, $n=30$, $\delta=0.99$, $d$ up to 120, C++ implementation). Shaded bands show $\pm 1$ standard error. Same panels as Figures~\ref{fig:deep_gaussian} and~\ref{fig:deep_qary}. The profile is intermediate between the Gaussian and $q$-ary regimes; Thermal-Adaptive picks $\alpha > 1$ and overtakes SS-GG from $d \ge 80$.}
\label{fig:deep_goldstein}
\end{figure}

\begin{figure}[t]
\centering
\includegraphics[width=\textwidth]{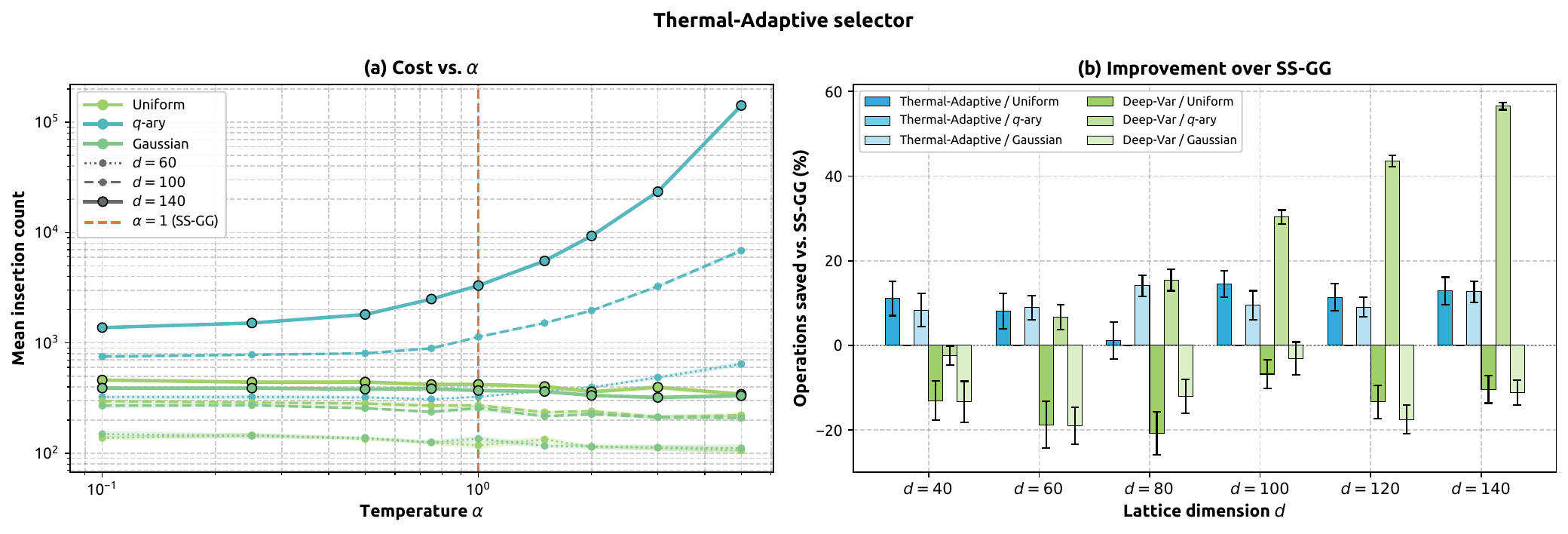}
\caption{Temperature spectrum and Thermal-Adaptive selector ($\delta=0.99$, $n=30$, C++ implementation). Panel~(a) sweeps fixed $\alpha$ at $d \in \{60,100,140\}$: Uniform and Gaussian inputs benefit from larger $\alpha$, while $q$-ary costs rise sharply beyond $\alpha \approx 1$, by roughly two orders of magnitude at $\alpha = 5$. Shaded bands are $\pm 1$ standard error; relative uncertainty is 4--8\% at each point, which is small on the three-decade log scale and renders as a narrow ribbon around each line. Panel~(b) shows percentage improvement over SS-GG across $d \in \{40,\ldots,140\}$, with whiskers denoting $\pm 1$ standard error. Thermal-Adaptive saves 8 to 15\% of operations on Uniform and Gaussian inputs at every tested dimension while matching SS-GG on $q$-ary by construction. Deep-Var follows a different pattern: on $q$-ary its saving over SS-GG grows with~$d$, from roughly $-2\%$ at $d{=}40$ to $+57\%$ at $d{=}140$, while on Uniform and Gaussian it stays near or below the SS-GG baseline with no clear trend.}
\label{fig:thermal_spectrum}
\end{figure}

\paragraph{G-DLLL.}
G-DLLL (Section~\ref{sec:gdlll}) targets equivalent-swap work rather than insertion count. We compare it against Deep-Var and SS-GG on the same benchmark suite, but the reading is different. G-DLLL is not the paper's main practical selector; it is a concrete test of the equivalent-work viewpoint behind Proposition~\ref{prop:roi}.

\begin{table}[t]
\centering
\caption{G-DLLL vs.\ Deep-Var and SS-GG: mean equivalent-swap count $W$
  and mean operation count $N$ ($\delta=0.99$, C++ implementation).
  $\Delta_W$ is the percentage reduction of G-DLLL over Deep-Var in $W$
  (total cascade depth).}
\label{tab:gdlll}
\scriptsize
\setlength{\tabcolsep}{1.5pt}
\begin{tabular}{llrrrrrr}
\toprule
Family & $d$ & Deep-Var $N$ & Deep-Var $W$ & SS-GG $W$ & G-DLLL $N$ & G-DLLL $W$ & $\Delta_W$ \\
\midrule
Gaussian &  40 &   77 &   743 &   736 &   223 &   692 & $+$7\%  \\
         &  80 &  201 &  2334 &  2360 &   674 &  2272 & $+$3\%  \\
         & 120 &  348 &  5061 &  5070 &  1287 &  4981 & $+$2\%  \\
         & 160 &  521 &  9155 &  9232 &  2162 &  9078 & $+$1\%  \\
\addlinespace
$q$-ary  &  40 &  160 &  3432 &  3298 &   465 &  3086 & $+$10\% \\
         &  80 &  518 & 18100 & 20561 &  2205 & 14982 & $+$17\% \\
         & 120 & 1074 & 46820 & 66993 &  5727 & 36603 & $+$22\% \\
         & 160 & 1877 & 86423 &169756 & 10922 & 63581 & $+$26\% \\
\addlinespace
Goldstein-Mayer
         &  40 &   61 &   504 &   572 &   143 &   440 & $+$13\% \\
         &  80 &  310 &  5331 &  7791 &  1126 &  4347 & $+$18\% \\
         & 120 &  781 & 20524 & 36790 &  3767 & 16203 & $+$21\% \\
         & 160 & 1660 & 51221 &125042 &  8306 & 37398 & $+$27\% \\
\bottomrule
\end{tabular}
\setlength{\tabcolsep}{6pt}
\end{table}

G-DLLL reduces equivalent-swap count $W$ by 1 to 7\% on Gaussian lattices, 10 to 26\% on $q$-ary, and 13 to 27\% on Goldstein-Mayer (Table~\ref{tab:gdlll}), with the largest reductions at the highest tested dimensions. Proposition~\ref{prop:roi} is built around exactly this regime, where insertions are long and the total cascade work concentrates on a few deep moves.

The wall-clock trade-off is unfavorable. Each deep insertion still pays the fixed cost of candidate scanning, size reduction, and GSO recomputation, while only the cascade update scales with depth. In the current implementation those fixed costs dominate, so the 3 to $5\times$ increase in insertion count overwhelms the saving in $W$. G-DLLL is therefore a selector motivated by the theory rather than by runtime. It makes the equivalent-work objective explicit and shows that minimizing $W$ is a different problem from minimizing runtime.

% ===================================================================
\section{Discussion}
\label{sec:discussion}
% ===================================================================

The experiments support a simple picture. At the adjacent-swap level, the T-transform law of Theorem~\ref{thm:variance} depends only on the local geometry of the violated Lov\'asz step, not on the input distribution, and the deep-insertion benchmarks of Section~\ref{sec:experiments} behave consistently with that local law across all three lattice families.

The data then split into two regimes. On flat profiles, especially Gaussian lattices, many deep-insertion candidates have almost the same SS score. Thermal-Adaptive improves operation counts by breaking those near-ties. Raising $\alpha$ above 1 steepens the gradient on the largest norms and turns a nearly blind selector into a directional one. That is why the gain appears in operation count without any visible loss in $\delta_0$ or final variance.

On $q$-ary lattices the story is different. A half-rank $q$-ary profile is bimodal, with roughly $d/2$ entries near $\tfrac{1}{2}\ln q$ and $d/2$ near zero, creating $\Theta(d^2/4)$ pairwise inversions and an $\Omega(d^2)$ floor for any adjacent-swap strategy. Deep insertions matter because they resolve $k-j$ inversions in one move. Within that regime, Deep-Var's log-space objective compresses the two plateaus and stops earlier, while SS-GG and Thermal-Adaptive preserve the full scale separation and continue to improve the basis.

DeepLLL, Pot-DeepLLL~\cite{fontein2014potlll}, SS-DeepLLL~\cite{yasuda2017analysis,yasuda2019new}, and the X-GG framework~\cite{bhattacherjee2025greedy} all pick a descent objective and optimize it greedily. The majorization viewpoint gives a common geometric explanation. Proposition~\ref{prop:canonical} sorts these objectives into two compatible classes: symmetric strictly Schur-convex scores, which contain the thermal family, and weighted separable scores with position-ordered weights, which contain Pot-DeepLLL and the mixed thermal-weighted hybrids. Within the thermal family, SS-GG is the unique point whose $r$-gradient is independent of position and magnitude (Corollary~\ref{cor:standard_scores}); Thermal-Adaptive moves $\alpha$ within that family and improves the flat-profile cases where that fixed endpoint is least informative. Proposition~\ref{prop:roi} gives a second consequence for total cascade work, and G-DLLL makes it algorithmic at the cost of exposing a fixed-cost barrier between selector quality and update cost. Appendix~\ref{app:variants} collects four probes of the design-space boundary. G-DLLL-RT, G-DLLL-CA, and the $f_{\alpha,\beta}$ hybrid slide along the G-DLLL $(W, \delta_0)$ Pareto point without pushing it outward, as Remark~\ref{rem:anchored_fail} predicts; Schur-$K$ hits the weak-Schur plateau and stalls. That boundary points toward look-ahead methods rather than further tuning inside the same myopic design space.

% ===================================================================
\section{Conclusions}
\label{sec:conclusions}
% ===================================================================

Each non-degenerate Lov\'asz swap is a T-transform on the GSO log-norm profile (Theorem~\ref{thm:variance}), so every Schur-convex functional of the profile is nonincreasing, with strict decrease for strictly Schur-convex functionals. When the normalized profile lies in the positive simplex, normalized entropies move in the opposite direction. This gives a local, distribution-independent law for LLL dynamics. Variationally, the arithmetic progression of slope $-c_\delta$ is the unique minimum-variance profile among all profiles whose adjacent gaps meet or exceed the worst-case Lov\'asz threshold (Proposition~\ref{prop:gsa_minvar}), so the slope of the GSA envelope follows from the geometry of the constraints rather than from a separate heuristic ansatz. That envelope is a worst-case upper bound on the steepness of LLL output; empirical LLL slopes~\cite{gama2008predicting} are flatter because LLL halts at the first violation-free profile it meets, which usually sits strictly inside the gap polytope rather than on its boundary (Remark~\ref{rem:anchored_fail}). Dissipatively, variance drops at each non-degenerate step by an amount fixed by the gap closure of the corresponding T-transform (Proposition~\ref{prop:convergence}).

The same geometry organizes deep-insertion selectors into two Lov\'asz-compatible classes (Proposition~\ref{prop:canonical}): symmetric strictly Schur-convex scores, which contain the thermal scores $\phi_\alpha = \sum r_i^\alpha$, and weighted separable scores with position-ordered weights, which contain Pot-DeepLLL~\cite{fontein2014potlll} and the mixed thermal-weighted hybrids. Within the thermal family, SS-GG~\cite{bhattacherjee2025greedy} is the unique member whose $r$-gradient is independent of position and magnitude. Thermal-Adaptive moves $\alpha$ within that family from the initial profile, and in our benchmarks it reduces operation counts by 8 to 15\% relative to SS-GG on Uniform and Gaussian inputs at every tested dimension while recovering SS-GG on $q$-ary by construction (Section~\ref{sec:experiments}). Proposition~\ref{prop:roi} motivates G-DLLL through the ROI efficiency $\eta(k,j) = \Delta V(k,j)/(k-j)$, and the resulting selector reduces equivalent-swap counts~$W$ by 1 to 27\% across the tested families, but not wall-clock time. Appendix~\ref{app:variants} gathers nearby variants; none dominates the main selectors, and together they sharpen the picture of Remark~\ref{rem:anchored_fail}: any score that pulls the profile toward the GSA target $\mathbf p^*$ can reject swaps that the Lov\'asz condition itself accepts, because LLL operates inside the gap polytope while $\mathbf p^*$ sits on its worst-case boundary. The variants slide along the existing Pareto frontier rather than crossing it.

Three directions stand out. The first extends the picture from variance to the entropy side. Normalized Shannon and R\'enyi entropies are Schur-concave on the simplex of nonnegative coordinates, so the T-transform reading transfers verbatim only when the profile is all-positive (Remark~\ref{rem:entropy}). Pre-LLL profiles routinely contain negative log-norms, and a Lov\'asz swap can flip the sign of a coordinate. An adaptive selector built from a Schur-concave entropy needs a normalization that keeps the simplex condition along the swap trajectory; we leave that construction open. The second direction extends to BKZ, where overlapping windows break the simple global majorization picture; algorithmically, update cost rather than score design becomes the bottleneck for nonlocal selectors. A natural near-term test is block-level thermal adaptation inside BKZ. Early blocks tend to have flatter local profiles, so a per-block $\alpha$ calibrated from the local profile would test whether the present local guarantees extend block-wise. The third direction is to move beyond myopic acceptance. A look-ahead variant comparing short rollouts from the top few candidates could test whether the near-degeneracy exploited by Thermal-Adaptive has a useful multi-step analogue, and whether it can sidestep the obstruction of Remark~\ref{rem:anchored_fail}, where myopic anchoring on $\mathbf p^*$ rejects Lov\'asz-valid swaps. Non-separable symmetric Schur-convex scores present a related open question: they need strict rather than weak descent to avoid stalling, and a promising candidate would strictly decrease under every T-transform on the \emph{unsorted} profile.

% ===================================================================
%  ACKNOWLEDGMENTS
% ===================================================================
\section*{Acknowledgments}

This work received financial support from the Spanish Government through grant PID2023-148716OB-C33, funded by MICIU/AEI/10.13039/501100011033, under the project ``DIstributed Smart Communications with Verifiable EneRgy-optimal Yields (DISCOVERY)''. Additional support was provided by the Community of Madrid through the research and development activities program, reference TEC-2024/COM-504 (acronym RAMONES-CM), awarded via Order~5696/2024 of the General Directorate of Research and Technological Innovation.

% ===================================================================
%  DATA AVAILABILITY
% ===================================================================
\section*{Data Availability}
\label{sec:data_avail}

The reference implementation, benchmark suite, and plot/table generators
are released as the open-source repository \texttt{variationaLLL} at
\url{https://github.com/perlab-uc3m/variationaLLL}. The canonical implementation
is the C++/\texttt{fplll} code in \texttt{src/}, which underlies the
deep-insertion benchmarks; the Python scripts in \texttt{scripts/}
reproduce the figures and tables of this paper.

% ===================================================================
%  DECLARATION OF INTEREST
% ===================================================================
\section*{Declaration of Interest}

The authors report no conflicts of interest.

% ===================================================================
%  REFERENCES
% ===================================================================
\printbibliography

% ===================================================================
%  APPENDICES
% ===================================================================
\appendix

% ===================================================================
\section{Profile Universality: Cross-Dimension Correlation}
\label{app:universality}
% ===================================================================

This appendix documents the empirical support for Remark~\ref{rem:schur}. We generated $n$ random lattices per dimension with i.i.d.\ $\mathrm{Uniform}(\{-10,\ldots,10\})$ entries, computed the normalized post-LLL log-norm profile $\hat{\mathbf{p}} = \mathbf{p} / \|\mathbf{p}\|_1$ for each, and measured cross-dimension correlation as the Pearson correlation between mean normalized profiles at two dimensions after interpolating to a common grid. The sweep covers $d \le 140$, matching the range of the deep-insertion benchmarks of Section~\ref{sec:experiments}.

\begin{table}[H]
\centering
\caption{Cross-dimension Pearson correlations of mean normalized post-LLL profiles. Sample sizes: $n = 600$ at $d{=}30$, $n = 360$ at $d{=}40$, $n = 180$ at $d{=}60$, $n = 120$ at $d{=}80$, $n = 90$ at $d{=}100$, $n = 60$ at $d{=}120$, $n = 45$ at $d{=}140$.}
\label{tab:universality}
\small
\begin{tabular}{rrrrrrrr}
\toprule
 & $d{=}30$ & $d{=}40$ & $d{=}60$ & $d{=}80$ & $d{=}100$ & $d{=}120$ & $d{=}140$ \\
\midrule
$d{=}30$ & ---   & 0.994 & 0.972 & 0.949 & 0.936 & 0.925 & 0.917 \\
$d{=}40$ &       & ---   & 0.991 & 0.975 & 0.964 & 0.955 & 0.947 \\
$d{=}60$ &       &       & ---   & 0.995 & 0.988 & 0.982 & 0.977 \\
$d{=}80$ &       &       &       & ---   & 0.998 & 0.995 & 0.991 \\
$d{=}100$&       &       &       &       & ---   & 0.999 & 0.997 \\
$d{=}120$&       &       &       &       &       & ---   & 0.999 \\
\bottomrule
\end{tabular}
\end{table}

For all pairs with $d \geq 30$, the minimum correlation is $0.917$ (at $d{=}30$ vs.\ $d{=}140$); for pairs with $d \geq 60$ the minimum rises to $0.977$. The decay with dimension gap is gradual: adjacent rows in the table all exceed $0.99$, and even the largest gap tested ($d{=}30$ to $d{=}140$) stays above $0.91$. The same universality holds across uniform, Gaussian, sparse, and $q$-ary inputs (Pearson $> 0.965$ at $d \geq 30$) and across $\delta \in \{0.55, 0.75, 0.90, 0.99\}$ (minimum cross-$\delta$ correlation $0.974$ at $d = 30$).

% ===================================================================
\section{Auxiliary Variants and Counterexamples}
\label{app:variants}
% ===================================================================

This appendix collects four probes of the Schur-convex design space. None dominates SS-GG, Thermal-Adaptive, or G-DLLL. Each variant targets one geometric feature of Proposition~\ref{prop:canonical} or Remark~\ref{rem:anchored_fail}, and each confirms that staying inside the Schur-convex family with myopic acceptance moves the operating point along the Pareto frontier rather than pushing it outward.

All four probes share the same protocol: $\delta = 0.99$, $n = 8$ lattices per $(family, d)$ cell, $d \in \{40, 80\}$, seed 42. The reduced sample size and dimension range reflect the exploratory role of these variants; they probe qualitative behavior, not whether the small differences scale to higher dimensions. Representative numbers appear in Table~\ref{tab:variants}.

\paragraph{C.1. Schur-$K$: a non-separable symmetric counterexample.}
Proposition~\ref{prop:canonical} leaves open the non-separable symmetric Schur-convex corner. The cleanest probe is the top-$K$ partial sum of the sorted log-profile, $\phi(\mathbf p) = \sum_{i=1}^{K} p_{[i]}$ with $K = \lceil d/2 \rceil$. The two-position update of Remark~\ref{rem:ssgg_inc} still applies, so the selector costs no more per candidate than SS-GG. Empirically it fails cleanly. On Gaussian $d{=}80$, Schur-$K$ stops after $26$ insertions instead of SS-GG's $173$, but $\delta_0$ worsens by $1.6 \cdot 10^{-3}$ and profile variance more than doubles. On $q$-ary $d{=}40$, $\delta_0$ worsens by $0.076$ and final variance jumps from $0.13$ to $5.17$. The issue is not speed; it is stalling. Partial sums of a sorted tuple are weakly, but not strictly, Schur-convex under the moves seen by the greedy loop. A T-transform that redistributes mass entirely inside the top-$K$ block, or entirely inside the bottom $d{-}K$ block, leaves $\phi$ unchanged. Such Lov\'asz-admissible deep insertions score zero descent and get rejected by the strictly positive acceptance rule, even though LLL would accept them. The word \emph{strictly} in Proposition~\ref{prop:canonical}(i) is essential.

\paragraph{C.2. G-DLLL-RT and G-DLLL-CA: residual shortlist and cost-aware selection.}
G-DLLL-RT restricts the G-DLLL shortlist to candidates with positive GSA residual $z_i = p_i - p_i^* > 0$, targeting the slab where the deviation from the asymptotic slope is largest. G-DLLL-CA divides $\Delta V(k,j)$ by a fixed per-candidate cost model rather than by geometric depth $k-j$, making the ROI efficiency explicitly cost-aware. On Gaussian inputs, both variants match G-DLLL's $\delta_0$ and equivalent-swap count $W$ at lower operation count $N$; for Gaussian $d{=}80$, $N$ drops from $646$ to $351$ at the same $W = 2230$. On $q$-ary $d{=}80$, G-DLLL-RT saves about $5\%$ of $N$ relative to G-DLLL-CA and shifts $W$ upward by about $5\%$, with $\delta_0$ staying within $3 \cdot 10^{-4}$. Neither variant crosses the Pareto frontier spanned by G-DLLL and Thermal-Adaptive. Anchoring the shortlist on $\mathbf p^*$ or on a fixed cost model breaks symmetry in exactly the way Remark~\ref{rem:anchored_fail} predicts, so the operating point slides along the frontier instead of pushing it outward.

\paragraph{C.3. The $f_{\alpha,\beta}$ hybrid.}
The weighted separable class of Proposition~\ref{prop:canonical} admits the two-parameter family
\begin{equation*}
f_{\alpha,\beta}(\mathbf{r}) = \sum_i (d - i + 1)^\beta\, r_i^\alpha, \qquad \alpha > 0,\ \beta \geq 0,
\end{equation*}
which interpolates between the thermal family at $\beta = 0$ and a Pot-DeepLLL-like position weighting at $\alpha = 1$. At $\beta = 2$, using Thermal-Adaptive's $\alpha$, equivalent-swap counts drop by $4$ to $8\%$ on Gaussian inputs but rise by $1$ to $7\%$ on $q$-ary inputs. So no fixed $\beta$ improves strictly on the $\beta = 0$ endpoint. A profile-adaptive $\beta$ schedule would still lie inside the weighted separable class, but would need a second calibration statistic beyond the $\mathrm{CV}$ used for $\alpha$, and the gain is bounded by that same $4$ to $8\%$ window.

\paragraph{C.4. Thermal-Sched: re-calibrating $\alpha$ during reduction.}
Thermal-Adaptive fixes $\alpha = \alpha_0$ once, from the initial $\mathrm{CV}$. A softer rule re-estimates $\alpha$ every $P$ accepted insertions from the current profile. On Goldstein-Mayer at $d = 80$, this improves output quality by about $8 \cdot 10^{-4}$ in $\delta_0$ and $20\%$ in final variance relative to SS-GG, but at roughly twice the equivalent-swap cost. The trade-off is almost flat in the period $P$: for $P \in \{20,40,80,160\}$, $W$ stays in $[14.7\mathrm{k},\,17.0\mathrm{k}]$ against SS-GG's $8.0\mathrm{k}$, while $\delta_0$ stays at $1.0128 \pm 3 \cdot 10^{-5}$ against $1.0137$. So the extra work comes from moving $\alpha$ as $\mathrm{CV}$ drops during reduction, not from a bad rescheduling choice. Calibrating once remains the practically useful $(W,\delta_0)$ point.

\begin{table}[H]
\centering
\caption{Auxiliary variants on representative cells ($\delta = 0.99$, 8 lattices per cell, mean values). The first three columns reproduce the main-text baselines; the last three give the auxiliary variants. $W$ is equivalent-swap count, $\delta_0$ the output Hermite factor, and $V$ the final profile variance.}
\label{tab:variants}
\small
\setlength{\tabcolsep}{4pt}
\begin{tabular}{llrrrrrr}
\toprule
Family & $d$ & SS-GG & Th.-Ad.\ & G-DLLL & G-DLLL-RT & G-DLLL-CA & Schur-$K$ \\
\midrule
\multicolumn{8}{l}{\emph{Equivalent swaps $W$}} \\
Gaussian & 40 & 714  & 667  & 636  & 636  & 636  & 237 \\
Gaussian & 80 & 2358 & 2233 & 2212 & 2231 & 2231 & 462 \\
$q$-ary   & 40 & 3264 & 3264 & 3068 & 3151 & 3063 & 1001 \\
$q$-ary   & 80 & 20383 & 20383 & 15025 & 15763 & 14889 & 2202 \\
\midrule
\multicolumn{8}{l}{\emph{Output $\delta_0$}} \\
Gaussian & 40 & 1.00620 & 1.00620 & 1.00620 & 1.00620 & 1.00620 & 1.01236 \\
Gaussian & 80 & 1.00401 & 1.00401 & 1.00403 & 1.00403 & 1.00403 & 1.00564 \\
$q$-ary   & 40 & 1.01405 & 1.01405 & 1.01388 & 1.01352 & 1.01416 & 1.09031 \\
$q$-ary   & 80 & 1.01410 & 1.01410 & 1.01532 & 1.01504 & 1.01548 & 1.04418 \\
\midrule
\multicolumn{8}{l}{\emph{Final profile variance $V$}} \\
Gaussian & 80 & 0.119 & 0.119 & 0.125 & 0.128 & 0.128 & 0.266 \\
$q$-ary   & 80 & 0.494 & 0.494 & 0.597 & 0.556 & 0.595 & 7.70 \\
\bottomrule
\end{tabular}
\end{table}

Taken as a whole, Table~\ref{tab:variants} shows the auxiliary variants either clustering around the G-DLLL point (G-DLLL-RT, G-DLLL-CA) or losing clearly to SS-GG in output quality (Schur-$K$). Schur-$K$'s drop in $W$ on $q$-ary inputs comes with a large loss in $\delta_0$ and variance, confirming the plateau argument of paragraph C.1. The remaining variants move along the frontier without crossing it, matching the slab picture of Remark~\ref{rem:anchored_fail}.

\end{document}